\documentclass{amsart}

\usepackage{graphicx}
\usepackage{subfigure}
\usepackage{amsthm,amsmath}

\newtheorem{thm}{Theorem}[section]

\theoremstyle{definition}

\theoremstyle{remark}

\numberwithin{equation}{section}

\newcommand{\Tr}{\mbox{tr}}

\begin{document}


\title{Numerical exploration of a hexagonal string billiard}


\author{Hans L. Fetter}
\address{Departamento de Matem\'aticas, Universidad Aut\'onoma
Metropolitana--Iztapalapa, M\'exico, D.F., M\'exico}
\email{hans@xanum.uam.mx}





\begin{abstract}
In this paper we are interested in the motion  of a ball inside a
   billiard table bounded by a particular smooth  curve.
 This table  belongs to a family of billiards  which can all be
  drawn  by a common process: the so--called gardener's string construction.
The classical elliptical billiard is,  of course, the foremost
member of this family. So it should  come as no surprise that our
hexagonal string billiard shares many basic properties with the
latter, but, on the other hand, also exhibits some essential
differences with it.
\end{abstract}


 \maketitle



\section{Background}\label{intro}
Let us  consider the motion of a point inside a plane billiard table
bounded by a closed convex curve. This point will  always move along a
straight line until it hits the boundary where it is reflected
according to the well known principle: the angle of reflection is
equal to the angle of incidence.

In order  to get a  billiard table with a sufficiently smooth boundary
we can use the familiar method known as a the gardener's or the string construction. 
First we need to choose some  convex polygon $K$  and then we proceed to 
 wrap a loop of inelastic string around it, pulling  the string  tight at a point $P$
  and then moving  this  point $P$ around. 
 Note that when this technique is applied  to  a closed line segment one
 obtains  an  ellipse. 
 Much  is known about the billiard problem inside an ellipse:
  see for instance the books by Chernov and Markarian\cite{Chernov} and Tabachnikov\cite{Tabachnikov},
and also the articles by Berry\cite{Berry1981},
Korsch and Zimmer\cite{Korsch2002} and Acquistapace\cite{Acquistapace1984} just to name a few.

Another choice  for the convex polygon which has received a fair amount of attention
 is that when  $K$  is an equilateral triangle.
 References include Hubacher  \cite{Hubacher1987}, Gutkin and Katok \cite{GutkinKatok1995},
 Gutkin and Knill\cite{GutkinKnill1996}, Turner \cite{Turner1982}.
Here again the result is  a convex domain  whose boundary is  comprised of  portions of ellipses. 

In   Fig. \ref{lab1} we illustrate the gardener's construction of a table $\Omega$ for the case when  $K$ is  a  regular hexagon.
The point $P$, as it  moves   around,  traces out a  curve  $\partial \Omega$ which again consists of a certain number of  elliptical arcs.

  Varying the length $l$ of the inelastic string one gets, of course, a whole family of billiard tables.
  We need to restrict our attention, however, to only one special value of $l$. This choice is dictated
  by the smoothness we get for the boundary $\partial \Omega$. 
   So it is for it only for this particular case 
  that we want to obtain a description, as complete as possible,  of all the  trajectories
as the billiard ball (point particle)  bounces elastically  off the walls. 
 
\begin{figure}[htp]
\centering
\includegraphics[height=2.4in,width=2.4in]{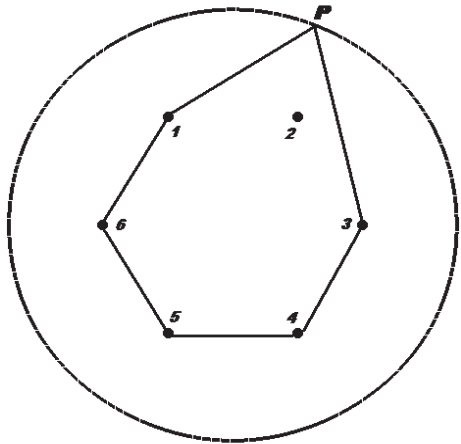}
\caption{Hexagonal string billiard}\label{lab1}
\end{figure}


\medskip

In closing this introduction let us just mention that the study of billiards and in
particular the subject of  billiards inside smooth convex curves  is a very  active 
area of research in the field of dynamical systems. The interested reader may consult
 some of the following books and  research articles dealing with
 various important aspects of this topic 
(see \cite{Korsch, Kozlov, Chernov, Tabachnikov}  and \cite{Amiran1988, GutkinKatok1995, GutkinKnill1996, Gutkin2003, Hubacher1987, Lazutkin1973, Mather1982} respectively).

\section{A family of billiards of class $C^{2}$}\label{results}

The billiard table obtained by the string construction
for an equilateral triangle, regardless of the string length used,  
is of class $C^{1}$ (see Hubacher  \cite{Hubacher1987}). 
In general, this is also true for all the other billiard tables obtained from regular polygons by the 
string construction, including the regular hexagon.
However, with a proper choice of the string length $l$ one can  
construct a table  having greater regularity.
To accomplish this  let us  first outline  the setup and also the strategy  behind the construction process for a  whole family of ``smooth'' string billiards. To start we select $K$ as  any  convex regular   $n$-gon ($n \geq 5$).
So right from the beginning we  exclude both the   equilateral triangle and the square from further considerations.
The reason for that is because we need to be able to  stellate $K$.
A stellated polygon can be derived from a regular polygon   by adding identical (congruent)
isosceles triangles to all its sides. 
Instead of adding them all we add only one of those triangles to our given polygon $K$.
The perimeter of this  resulting polygon is the sought for special value for
 the  string length $l$. Shortly we shall  find an explicit expression for it.
 A portion of a regular polygon having vertices   $F_{1},F_{2},\ldots,F_{n}$ with side length 2
 is shown in  left part of Fig.~\ref{lab2}.
 The coordinates of the vertices, in terms of the interior angle  $\alpha = \frac{n-2}{n}\pi$,
 are as indicated below:
 
\begin{align*}
F_{1}  &= [-1, 0] \\
F_{2}  &= [1, 0]\\
F_{3}  &= [1-2\cos \alpha, -2\sin \alpha]\\
F_{4}  &= [-1-2\cos \alpha+4 \cos^2\alpha, 2(-1+2\cos \alpha)\sin \alpha]\\
\ldots \\
F_{n-1}  &= [1+2\cos \alpha-4 \cos^2\alpha, 2(-1+2\cos \alpha)\sin \alpha] \\
F_{n}  &= [-1+2\cos \alpha, -2\sin \alpha]
\end{align*}

Now let us add the  isosceles triangle with top vertex $G_{1}$ to the side $F_{1},F_{2}$ (see
left part of Fig.~\ref{lab2}).
 In this  figure  we have actually  added several  identical (congruent) isosceles triangles to the  sides
 of the regular   $n$-gon. They have vertices $G_{1},G_{2}, \ldots, G_{n}$, whose
 coordinates   are given  below:

\begin{align*}
G_{1}  &=  \left[0, - \frac{\sin \alpha}{\cos \alpha} \right]\\
G_{2} &= \left[\frac{\cos \alpha -1}{\cos \alpha}, 0\right] \\
\ldots \\
G_{n} &= \left[-\frac{\cos \alpha -1}{\cos \alpha}, 0\right]
\end{align*}

Each of the equal sides of the  isosceles triangle  $F_{1}F_{2}G_{1}$  has length $ d = \left|F_{1} G_{1} \right| =   \left|F_{2} G_{1} \right| =   -1/\cos \alpha.$ 

 So finally we get the expression for the particular string length  we were looking for:
  $$l = 2(n-1) + 2 d = 2 (n-1) - \frac{2}{\cos(\frac{n-2}{n} \pi)}$$

\begin{figure}[htp]
\centering
\includegraphics[height=2.8in,width=6.0in]{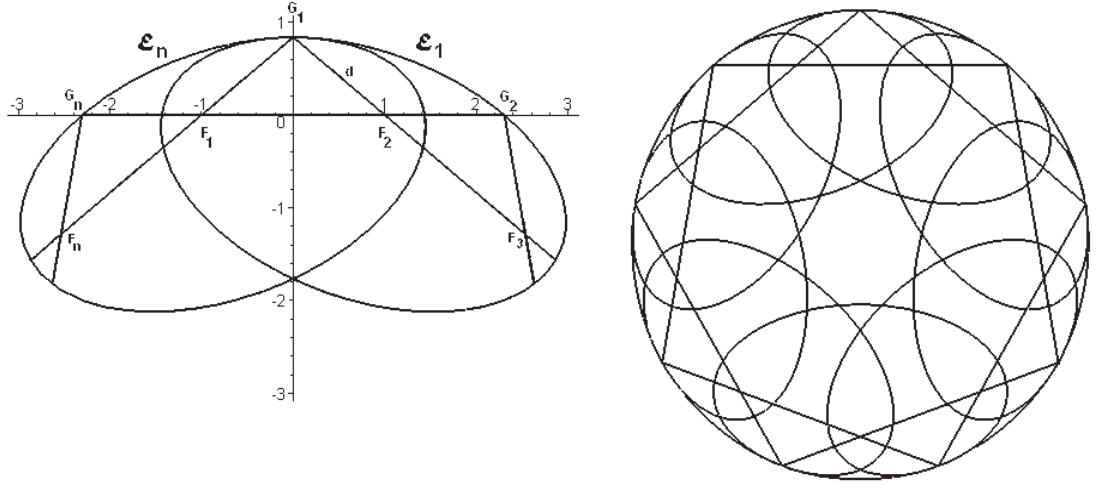}
\caption{Construction of  $n$-gonal string billiard  and
enneagonal string billiard }\label{lab2}
\end{figure}


So much for the setup. In the right part of Fig.~\ref{lab2} we illustrate the result of
the construction process described above for the case when  $K$ is  a   regular nine-gon.
That we get billiard tables with greater regularity than just  $C^{1}$ is the contents of 
the following theorem.

\begin{thm}\label{gardener}
The generalized gardener's construction with length of the string   $$l = 2 (n-1) - \frac{2}{\cos(\frac{n-2}{n} \pi)}$$
provides  us with  a  family
of billiard tables (one for each convex regular   $n$-gon, $n \geq 5$) which are globally $C^{2}.$
\end{thm}

For the proof see ~\ref{A1}.

 
 Some immediate consequences from this  construction process include  the following:
 \begin{itemize}
 \item  As opposed to the general case we need to concern ourselves with only one  kind of ellipse.
 \item  The   closed boundary curve we get is a smooth union of $n$ elliptical arcs. 
  In the general case one gets twice that  number.
 \item  Since an ellipse will always have strictly positive curvature
   the   curvature $\kappa$ is a continuous, strictly positive function on $\partial \Omega$.
 \end{itemize}   


From this point  on we shall discuss one specific billiard only:
the hexagonal string billiard.
We shall now  provide some specific details for it (see Fig.~\ref{lab3}).
Let $K$ be the regular hexagon with sidelength 2.
Its vertices  have the coordinates 
  
$$\begin{array}{ccccccccc}
  F_{1}&=&(-1,\sqrt{3}),& F_{2}&=&(1,\sqrt{3}),& F_{3}&=&(2,0) \\
  F_{4}&=&(1,-\sqrt{3}),& F_{5}&=&(-1,-\sqrt{3}),& F_{6}&=&(-2,0)
\end{array}$$
Since  $n=6$,  $\alpha=\frac{2 \pi}{3}$ and $d = -1/\cos \alpha = 2$ 
 the  length of the string  $l = 2 (n-1) + 2 d = 14$.
Each of the six elliptical arcs uses a unique  pair of foci $F_{i},F_{j}$,
 where $i$ and $j$ are either  both even or both odd. 
To reflect this in our notation,  we have used arc $\stackrel{\frown}{ij}$
 for an arc focused at  $F_{i},F_{j}$.
 
 The equation for  the arc $\stackrel{\frown}{24}$ is 
$$ \frac{(x-1)^{2}}{6} + \frac{y^{2}}{9} = 1$$
   whereas the equation for the arc $\stackrel{\frown}{13}$
 is given by
 $$11y^{2}+2\sqrt{3}(x-6)y+9x^{2}-12x=60$$
 At the common point $(3,\sqrt{3})$
 the first derivative for  both yields
$-\sqrt{3}$ while  the second derivative is $-\frac{3}{2}\sqrt{3}.$

As for the curvature $\kappa$ of any elliptical arc, a simple calculation gives:
 $$\frac{\sqrt{6}}{9} \leq \kappa \leq \frac{3\sqrt{3}}{16}$$

 \begin{figure}[htp]
\centering
\includegraphics[height=2.2in,width=3.2in]{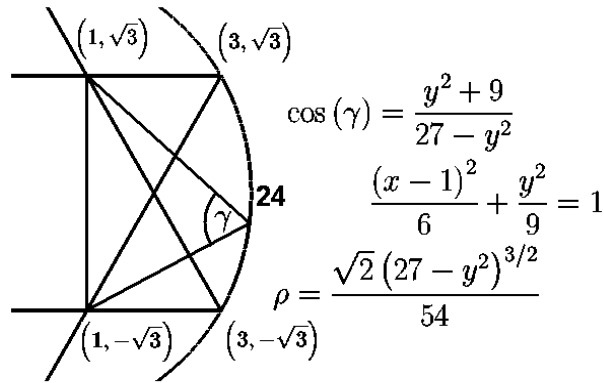}
\caption{Data of hexagonal  billiard}\label{lab3}
\end{figure}

\section{Types of orbits}\label{orbits}

The orbits of billiards are sequences of line segments on the
billiard table  with the segments corresponding to the straight
paths which the ball must follow  within the region until it hits
the boundary where it reflects elastically. Each such orbit is
determined by the choice of the  initial location  for the  billiard
ball as well as the  direction for the shot.

It is well-known that in  a generic or typical billiard one can
identify different types of trajectories: 
periodic, quasi-periodic (librational), 
 whispering-gallery (rotational) and chaotic.
We shall try to describe and illustrate  them
in the context of  the hexagonal string billiard.

\subsection{Periodic orbits}
In almost any study on billiards initial interest centers on
periodic orbits: those in which the particle follows the same path
endlessly.
In general these paths can be recognized easily because they simply
correspond to inscribed simple polygons or  star polygons.
When, for example, a convex billiard table has two perpendicular
symmetry axes then we immediately get  a period four orbit (Hasselblatt and Katok~\cite{Hasselblatt}).
Our billiard has  six axes of symmetry, so we get twelve points from their intersection
 with the boundary of the billiard table.
 Selecting arbitrarily any of these  points
 as the initial point   and then connecting
   it to the next one by skipping over exactly $k - 1 = 0,\ldots,5$
   intermediate points  we are bound to get a closed figure.
   Since this figure is always invariant under the symmetries,
 it follows that at each point of intersection with the table the angle of incidence
is equal to the angle of reflection.
Now, according as $k = 1,\ldots ,6$ the corresponding  periodic
orbits will close onto themselves after $\frac{12}{(12,k)}$ bounces, respectively, where
 $(n,k)$ denotes the greatest common divisor of $n$ and $k.$
 Some of these basic periodic orbits for our billiard are shown in 
  Fig.~\ref{lab4}.



\begin{figure}[htp]
  \begin{center}
    \subfigure[Unstable orbits of period 2, 3, and 6]{\label{lab4_a}\includegraphics[scale=0.34]{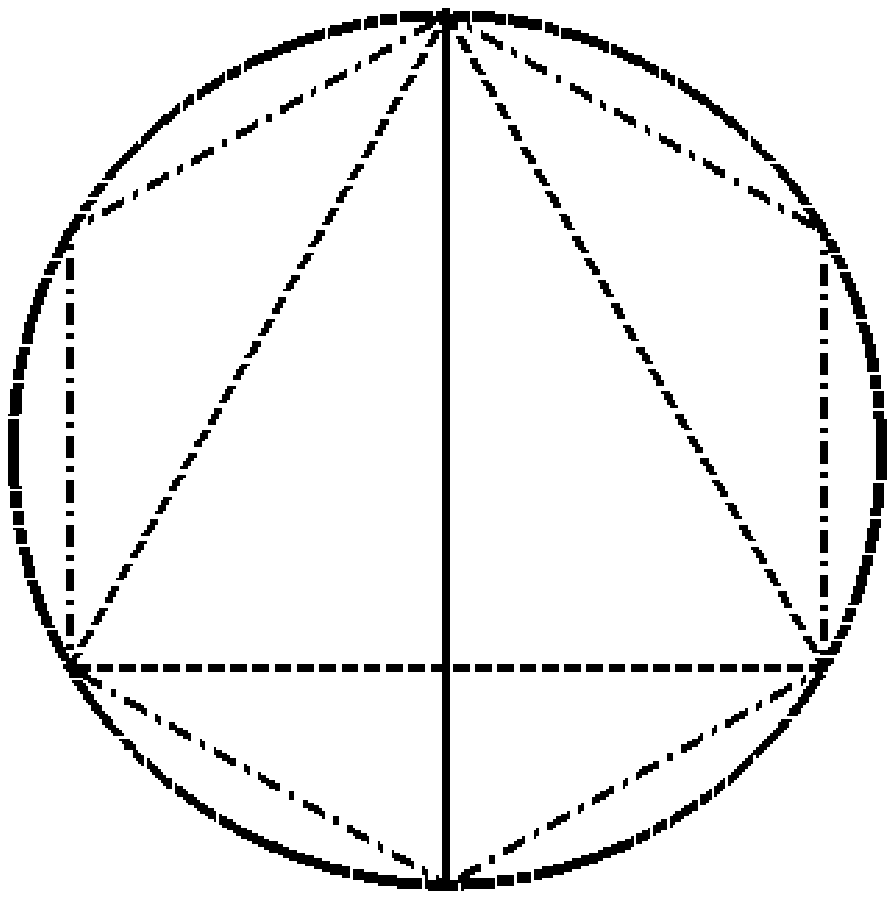}} 
    \subfigure[Stable orbits of period 2, 3, and 6]{\label{lab4_b}\includegraphics[scale=0.34]{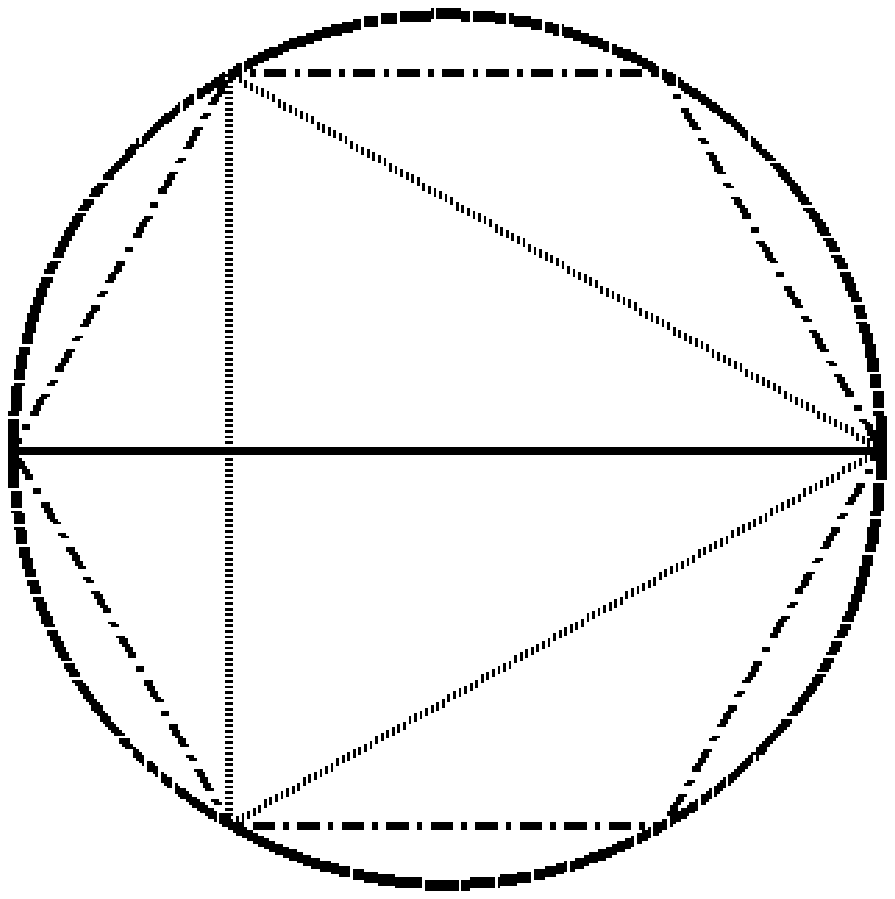}}
    \subfigure[Unstable orbits of period 4, 12, and 12] {\label{lab4_c}
    \includegraphics[scale=0.34]{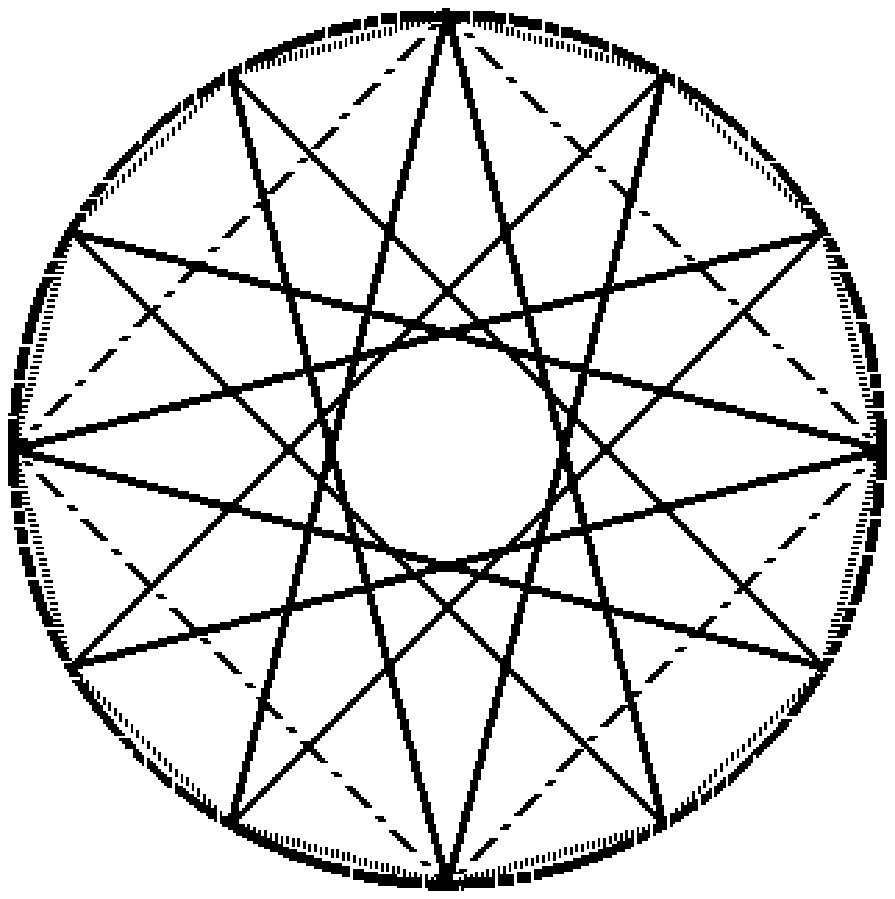}}
  \end{center}
  \caption{Periodic orbits ${12 \brace k}$ where $k=1\ldots6$ }
  \label{lab4}
\end{figure}


  
When dealing with periodic trajectories a very natural question is: how
many are there performing $n$  bounces and at the same time going around the boundary $k$ times?
Birkhoff~\cite{Birkhoff}   provided an astonishing answer to the
previous question as early as 1927 (see also
Tabachnikov \cite{Tabachnikov}):

 \begin{thm}
 For any $n \geq 2$ and every $k < \frac{n}{2}$, relatively prime,
 there exist two geometrically distinct $n-$periodic trajectories
 with the rotation number $k$.
 \end{thm}
 
 So, for example,
 the 2-,3- and 6-bounce orbits  in Fig.~\ref{lab4}(a) and those in Fig.~\ref{lab4}(b)
 are geometrically distinct. They are similar though, with a similarity ratio 
 of $$\frac{2 \sqrt{3}(-1+\sqrt{6})}{5} = 1.0042$$
 
 From Birkhoff's  theorem we are  lead  to the conclusion that there are always at least
 $\varphi(n)$  distinct $n-$periodic trajectories, where $\varphi(n)$ denotes Euler's  totient function.
  Thus when $n=17$ we must have  $\varphi(17)= 16$  distinct $17-$periodic trajectories.
 The reader can appreciate  the eight unstable 17-sided  star polygons that we  
  constructed for the hexagonal string billiard in Fig.~\ref{lab5}.

 \begin{figure}[htp]
\centering
\includegraphics[height=3.0in,width=6.0in]{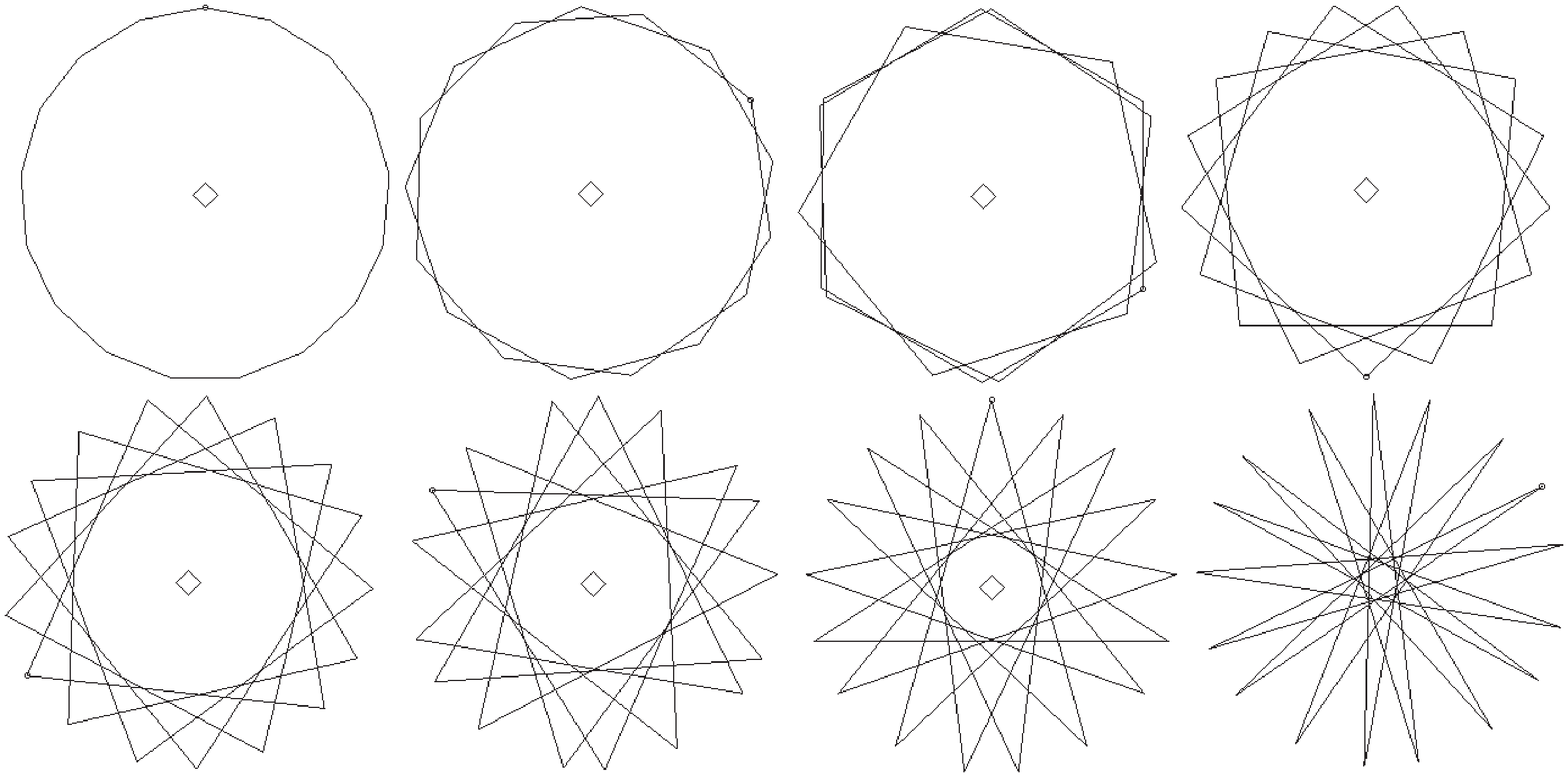}
\caption{Star polygons: ${17 \brace 1} \cdots {17 \brace8}$}\label{lab5}
\end{figure}

 It is clear that, because our billiard  has six-fold symmetry,  we can rotate 
  these  geometrically distinct $n-$periodic trajectories by certain angles.
 The  order of rotational symmetry is given by $\frac{6}{(6,n)}$ and 
 so this yields a grand total of
      $$\frac{6}{(6,n)}\varphi(n) $$
   $n-$periodic trajectories  for each $ n \geq 3.$

The first few values of this expression  are given in the following 
  Table~\ref{tab:mytable}.


  \begin{table}[htp]
  \begin{center}
    \begin{tabular}{|c|c||c|c|}
    \hline
    $n$ & $\frac{6}{(6,n)}\varphi(n)$ & $n$ & $\frac{6}{(6,n)}\varphi(n)$\\
    \hline
    3 & 4 & 8 &  12 \\
    4 & 6 & 9 & 12 \\
    5 & 24 & 10 &  12\\
    6 & 2 & 11 &  60\\
    7 & 36 &  12 & 4 \\
   \hline
   \end{tabular}
  \caption{Total number of $n-$gons.}
  \label{tab:mytable}
  \end{center}
  \end{table}


\subsection{Quasi-periodic orbits}
Sometimes one also  encounters
certain orbits which explore only  restricted parts (segments)  of the boundary $\partial \Omega$.
Every (stable) simple polygon or  star polygon (the two-bounce orbit included) gives rise to
 an infinite number of these quasi-periodic orbits. They resemble the underlying polygons: it is
 as if these just fattened  up. So that's why we will also refer to them
simply as  fat (star) polygons.
Orbits which explore $2, 3, 4, 6$ and $9$ parts can be seen  in Figs.~\ref{lab6} and  ~\ref{lab7}.
More examples can be found in  Figs.~\ref{lab11},  ~\ref{lab15}, ~\ref{lab23} and  ~\ref{lab25}.

\begin{figure}[htp]
\centering
\includegraphics[height=2.0in,width=2.0in]{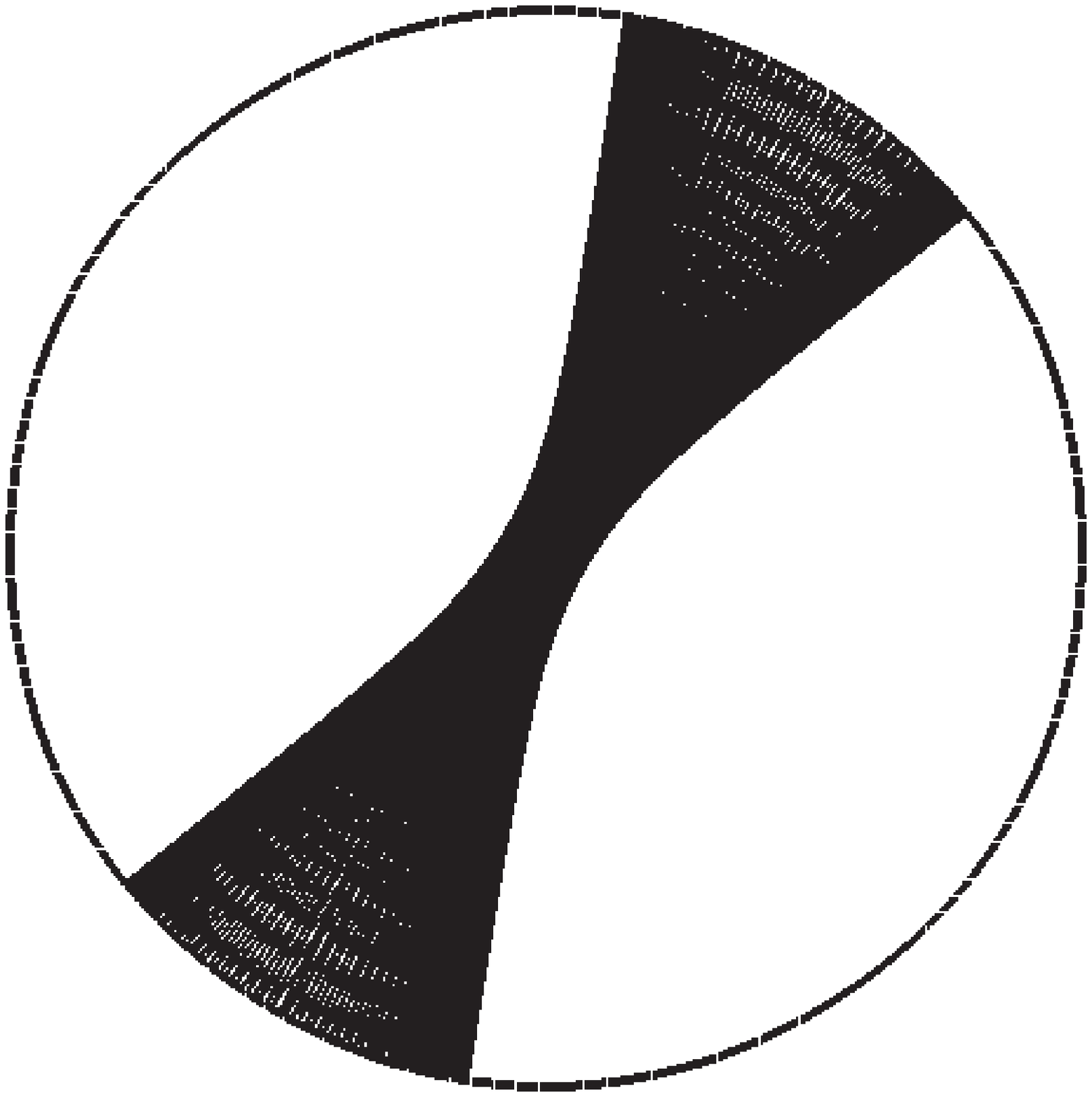}
\caption{A typical bouncing orbit or fat polygon of type ${2 \brace 1}$}\label{lab6}
\end{figure}

\begin{figure}[htp]
\centering
\includegraphics[height=4.0in,width=4.0in]{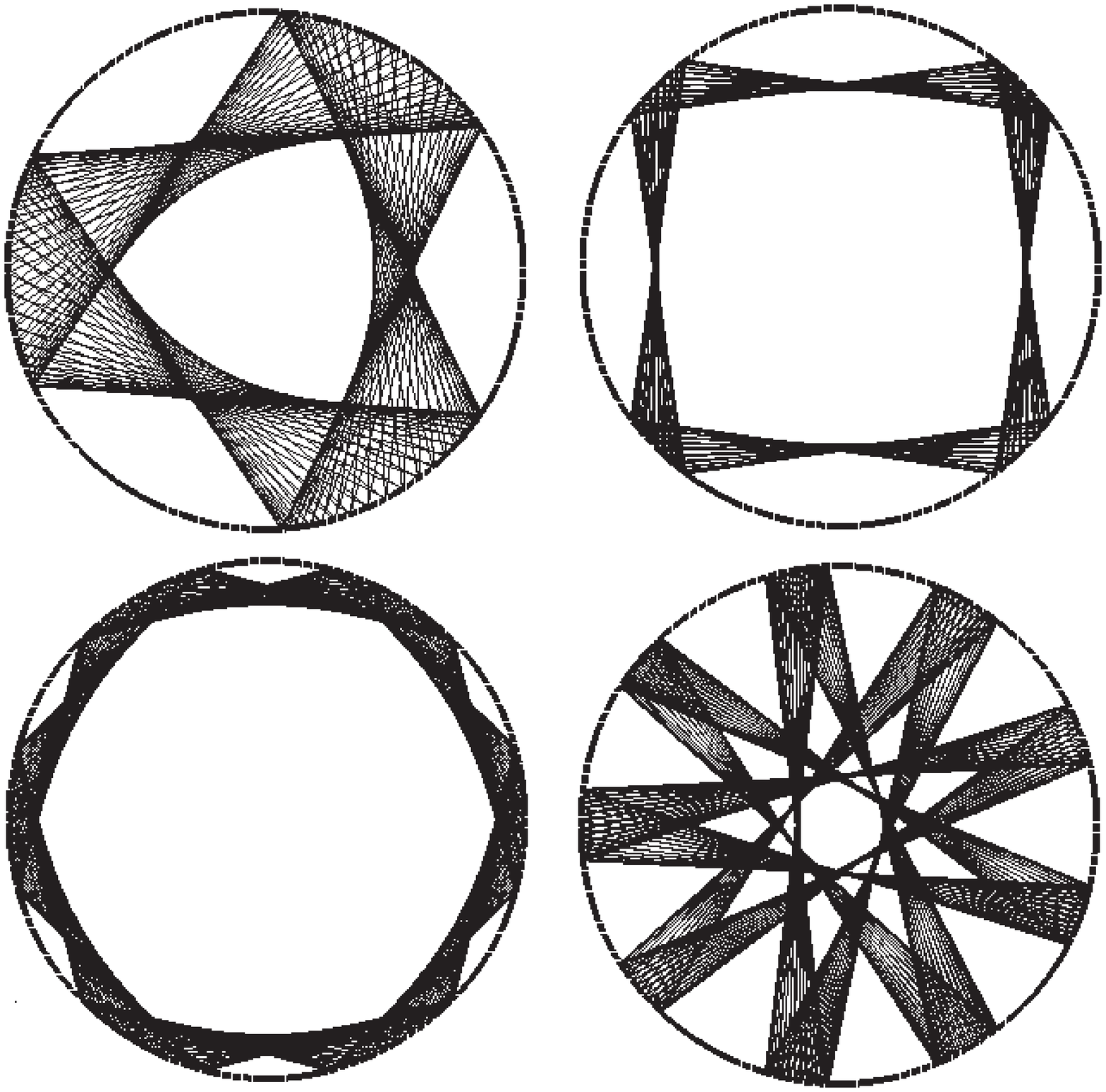}
\caption{Quasi-periodic orbits or
         fat polygons of type: ${3 \brace 1}, {4 \brace 1}, {6 \brace 1}$ and ${9 \brace 4}$ }\label{lab7}
\end{figure}






\subsection{Whispering gallery orbit}
 Another kind of orbits, which on occasions are also
referred to as  \emph{skipping} trajectories (see Berry
\cite{Berry1981}, Chernov and Markarian\cite{Chernov}), are
very common for most billiards.
They are characterized by the
fact that they bounce all round $\partial \Omega$ densely filling a
ring-like shaped region.

There are many  of these orbits in  our billiard
and we show several of them in  Fig.~\ref{lab8}.

\begin{figure}[htp]
\centering
\includegraphics[height=2.0in,width=6.0in]{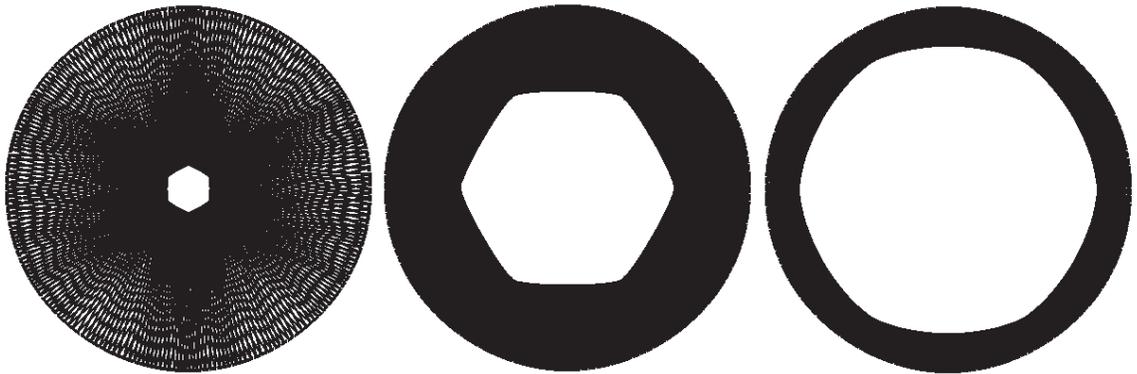}
\caption{Whispering gallery type orbits}\label{lab8}
\end{figure}


\subsection{Chaotic orbit}
Most  billiards exhibit chaotic behavior to some extent (circles and ellipses  excluded).
Chaotic trajectories look disordered, erratic and there is  no discernible pattern.
They usually manifest themselves more clearly in the surface of section,
which we will introduce later.
For the moment let us just say that there is no evidence that our billiard
has  these kinds of orbits.

\section{Some  analogous properties  of  the elliptical and the hexagonal string  billiard}

One of the most important properties of elliptic billiards is the following result which
essentially tells us that we can group the trajectories into three  families.

A trajectory for the billiard inside an ellipse either   
\begin{enumerate}
  \item   always passes through the two foci alternately (focal trajectory),  or
  \item   always intersects the open segment between the two foci (inner trajectory), or 
  \item   never intersects the closed segment between the foci (outer trajectory).
  \end{enumerate}

(See Siburg~\cite{Siburg}, Tabachnikov~\cite{Tabachnikov}, Chernov and Markarian~\cite{Chernov}).

In a completely analogous fashion for   the hexagonal string  billiard
we also have three  families of trajectories: focal, inner,  and outer.

\begin{thm}\label{ppal}
Consider a trajectory issuing from the point $P_{0}$ inside a
hexagonal string billiard  with vertices $P_{1},P_{2},\ldots.$

\begin{enumerate}
 \item[i)] If the segment $P_{0}P_{1}$ of the billiard trajectory
 is on a supporting line of $K$,
 then every segment $P_{i}P_{i+1}$ will be on a supporting line of
 $K$.
 
\item[ii)] If the segment $P_{0}P_{1}$ of a billiard trajectory does
 intersect the hexagon $K$ then
 all segments $P_{i}P_{i+1}$ will intersect it.

 \item[iii)]  If the segment $P_{0}P_{1}$ of a billiard trajectory does not
 intersect the hexagon $K$ then
 no segment $P_{i}P_{i+1}$ will intersect it.

 \end{enumerate}

\end{thm}

A sketch for the proof of this result can be found in ~\ref{B1}.

Another important result for the billiard inside an ellipse states 
that a billiard trajectory through the foci converges to the major axis of the ellipse.
 (See for example
Tabachnikov~\cite{Tabachnikov}, Batschelet~\cite{Batschelet1948}, Frantz~\cite{Frantz1994}, 
 Hasselblatt and Katok~\cite{Hasselblatt}, Wilker~\cite{Wilker1995}, Moser and Zehnder~\cite{Moser}).
    
 In practice this means that  any focal orbit   quickly becomes indistinguishable from a repeated
 tracing of the major axis of the ellipse.


Whereas a focal orbit in the case of the ellipse  converges  to the ``to and fro'' orbit
(the unstable two-bounce orbit), in that of  the hexagonal string  billiard it
 converges  to either one of the  equilateral triangles (the  unstable three-bounce orbits).
 
Before stating the result in more precise terms  we need an auxiliary result which provides 
bounds for the  sizes of the angles between incoming and outgoing rays  in focal orbits.
(See   Fig.~\ref{lab3}).

\begin{thm}\label{angle}
For a focal orbit of the hexagonal string billiard
the angle between incoming and outgoing segments  is
limited to the range
$$ 60^{\circ} \leq  \gamma \leq 70.53^{\circ}$$
\end{thm}

Proof in ~\ref{C1}.


Without loss of generality let  us choose the  three-bounce orbit shown in Fig.~\ref{lab4}(a).
The  vertices for this triangle are
$S_{1}=(0,2\sqrt{3})$, $S_{3}=(3,-\sqrt{3})$, and 
$S_{5}=(-3,-\sqrt{3})$.
 Consider a billiard trajectory, passing through the focus  $F_{2}$,
whose initial point $P_{0}$ is on    arc $\stackrel{\frown}{24}$.
See  Fig.~\ref{lab10}. 
Observe that, for example,  point $P_{4}$ is closer to
point $S_{1}$ than $P_{1}$, and then $P_{7}$ is even closer to it
and so on. Already we find that $P_{7}P_{8}P_{9}$ resembles triangle
$S_{1}S_{5}S_{3}$ very much.

Under these circumstances we have:
\begin{thm}\label{triangle}
If a billiard trajectory travels along a supporting line of $K$, so that it
goes alternately   through the focal points
$F_{2},F_{6},F_{4}$  then its  trajectory
 will quickly become indistinguishable from a repeated
 tracing of the triangle $S_{1}S_{5}S_{3}$.
\end{thm}

See ~\ref{D1} for proof.

\begin{figure}[htp]
  \begin{center}
    \subfigure[Trajectory  approaching the triangle  $S_{1}S_{5}S_{3}$]{\label{lab10_a}\includegraphics[scale=0.30]{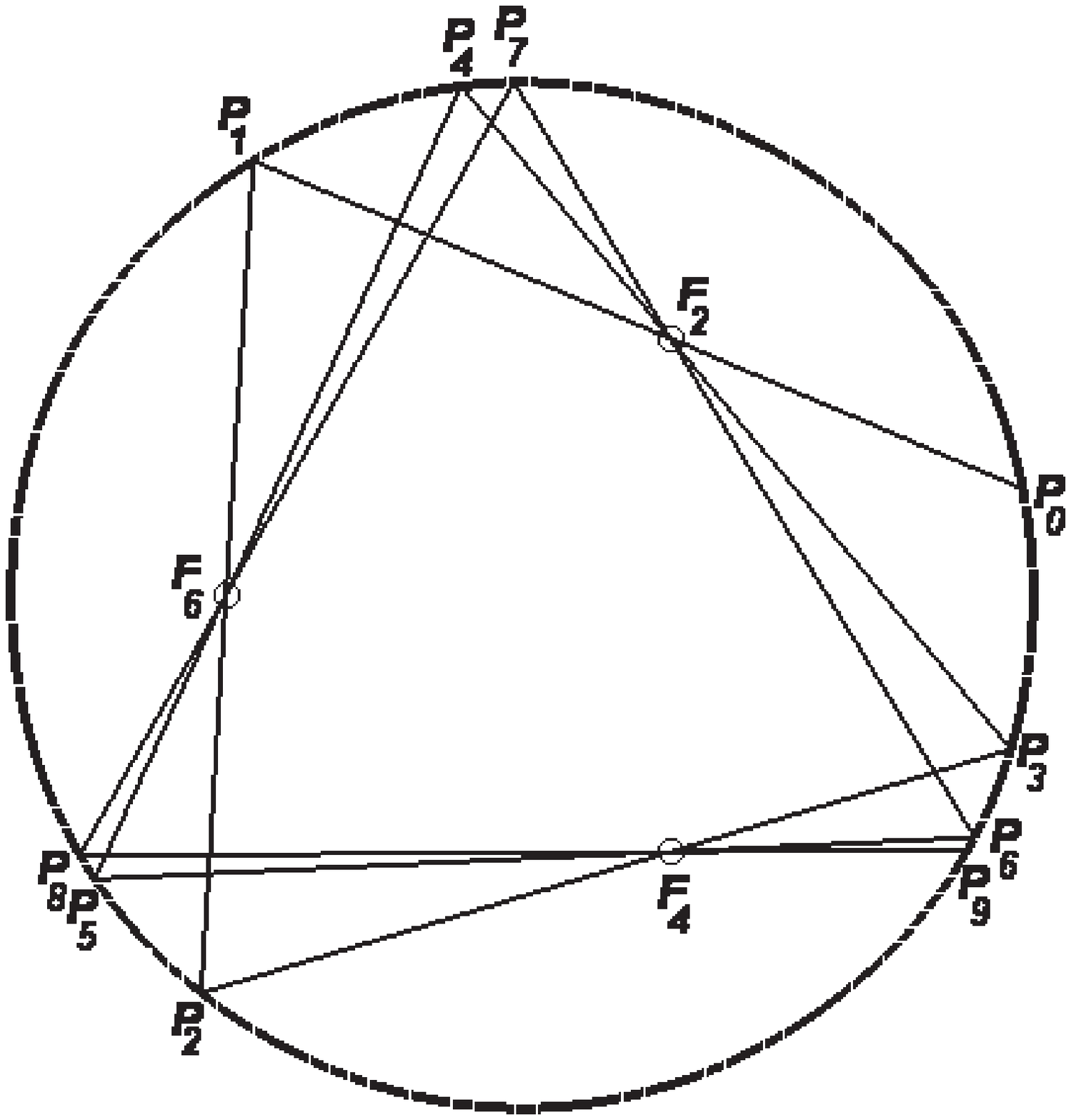}} 
   \subfigure[$\{s_{i}\},\{\varphi_{i}\},\{\alpha_{i}\}$]{\label{lab10_b}\includegraphics[scale=0.30]{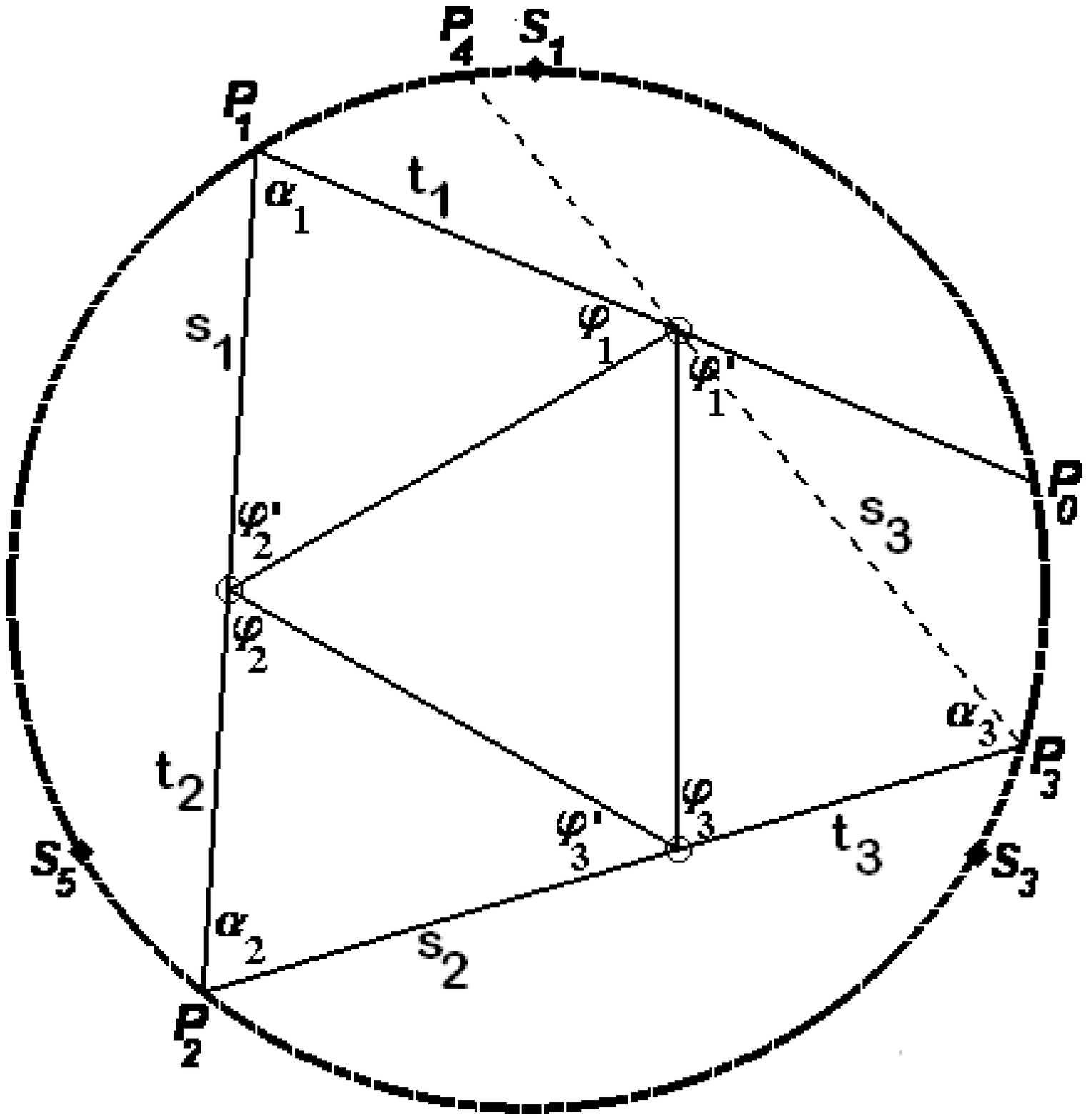}}
  \end{center}
  \caption{Focal orbits}
  \label{lab10}
\end{figure} 

 
\section{Stability}

Several authors include in
their work a brief discussion on the stability of certain
$n-$periodic orbits through the use of the so called deviation
matrix $M$.
 We also want to mention something about stability. For this purpose
we have chosen the  ${12 \brace 5}$ orbit, which can be  seen on the
top left part of Fig.~\ref{lab11}. This orbit is  not the same as
the earlier one we came across in Fig.~\ref{lab4}(c) (unstable ${12 \brace 5}$ orbit).

So following Berry \cite{Berry1981}, Lopac et al.
\cite{Lopac2006}, \cite{Lopac2001}, and Robnik
\cite{Robnik1983}   we have  to construct the deviation
matrix $M$ corresponding to the selected orbit:

 $$M = M_{1,2}M_{2,3}\cdots M_{12,1}$$
where each $2\times 2$ matrix $M_{i,k}$ can be expressed as

\[ M_{i,k} = \left(\begin{array}{cc}
-\frac{\sin \alpha_{i}}{\sin \alpha_{k}}+
\frac{\rho_{ik}}{R_{i}\sin \alpha_{k}} & -\frac{\rho_{ik}} {\sin \alpha_{i}\sin \alpha_{k}}  \\
-\frac{\rho_{ik}}{R_{i}R_{k}}+\frac{\sin \alpha_{k}}{R_{i}}
+\frac{\sin \alpha_{i}}{R_{k}} & -\frac{\sin \alpha_{k}}{\sin
\alpha_{i}}+ \frac{\rho_{ik}}{R_{k}\sin \alpha_{i}}
 \end{array}\right).\]

The symbol $R_{i}$ ($R_{k}$) denotes the radius of curvature at the
point of impact $P_{i}$ ($P_{k}$), $\alpha_{i}$ ($\alpha_{k}$) the
angle which the departing ray makes with the tangent on the billiard
boundary at the point of impact $P_{i}$ ($P_{k}$) and $\rho_{ik}$
the length of the path between two consecutive impact points $P_{i}$
and $P_{k}$.

For our particular choice we only get two kinds of these $2\times 2$
matrices $M_{i,k}$ since only the path lengths change between
impacts, all other data remaining constant.

In fact, if we define
$$\tau = \frac{3 \sqrt{\left(4947-2328 \sqrt{3}\right)}}{97}$$
then we get

\begin{enumerate}
 \item{}
$ R_{i} = R_{k} = \frac{ \sqrt{2}}{54} \left(\frac{2160+ 216
\sqrt{3}}{97}\right)^{3/2}$
 \item{}
$ \alpha_{i} = \alpha_{k} = 5 \pi/12,$ so that $ \sin (\alpha_{i}) =
\sin (\alpha_{k}) = \frac{\sqrt{2}(1+\sqrt{3})}{4} $
\item{}the two values which alternate

\begin{itemize}
\item[]  $\rho_{ik} =  \sqrt{\left(\frac{43}{3} \tau^2+8 \sqrt{3} \tau^2 +8
\tau+6 \sqrt{3} \tau+3 \right)}$
\item[] $\rho_{ik} = \frac{8}{3} \tau+\frac{4 \sqrt{3}}{3} \tau +2$
\end{itemize}

 \end{enumerate}
So
$$ M_{1,2} = M_{3,4} = M_{5,6} = M_{7,8} = M_{9,10} = M_{11,12} = T$$
$$ M_{2,3} = M_{4,5} = M_{6,7} = M_{8,9} = M_{10,11} = M_{12,1} = S$$
where
$$T = \left( \begin{array}{cc} 0.9830565623814575557 & -7.1795378524870580504  \\
             0.00467993843741796985  & 0.9830565623814575557
\end{array} \right)$$
and
$$S = \left( \begin{array}{cc} 0.9700724323780286883 & -7.1325295852446540944 \\
 0.00826627849706319570 & 0.9700724323780286883
 \end{array} \right)$$

Therefore our deviation matrix $M$ turns out to be

 $$M =  (T S)^{6} = \left( \begin{array}{cc} -.87488446006857883011 & -17.010433690097981202 \\
  .015415741802173798979 &  -.84327883276452294763
\end{array} \right)$$

Orbital stability depends on the eigenvalues of $M$.
 In fact, if by $\Tr(M)$ we denote the trace of $M$ then
 it is known that (Berry \cite{Berry1981}) we have the
 following possibilities:

 \begin{enumerate}
 \item{} if $|\Tr  (M)| < 2$ then the orbit is stable
 \item{} if $|\Tr  (M)| > 2$ then the orbit is unstable
 \item{} if $|\Tr  (M)| = 2$ then the orbit is neutrally stable.
\end{enumerate}

  For  our case we simply get that
$$ \Tr (M) = -1.7181632928331017777$$
which guarantees the stability of our chosen orbit.
 This analytically verified stability can also be observed
numerically: see Fig.~\ref{lab11}, where in addition to the
 ${12 \brace 5}$ polygon  we show several quasi-periodic orbits.
 
 \begin{figure}[htp]
\centering
\includegraphics[height=4.0in,width=4.0in]{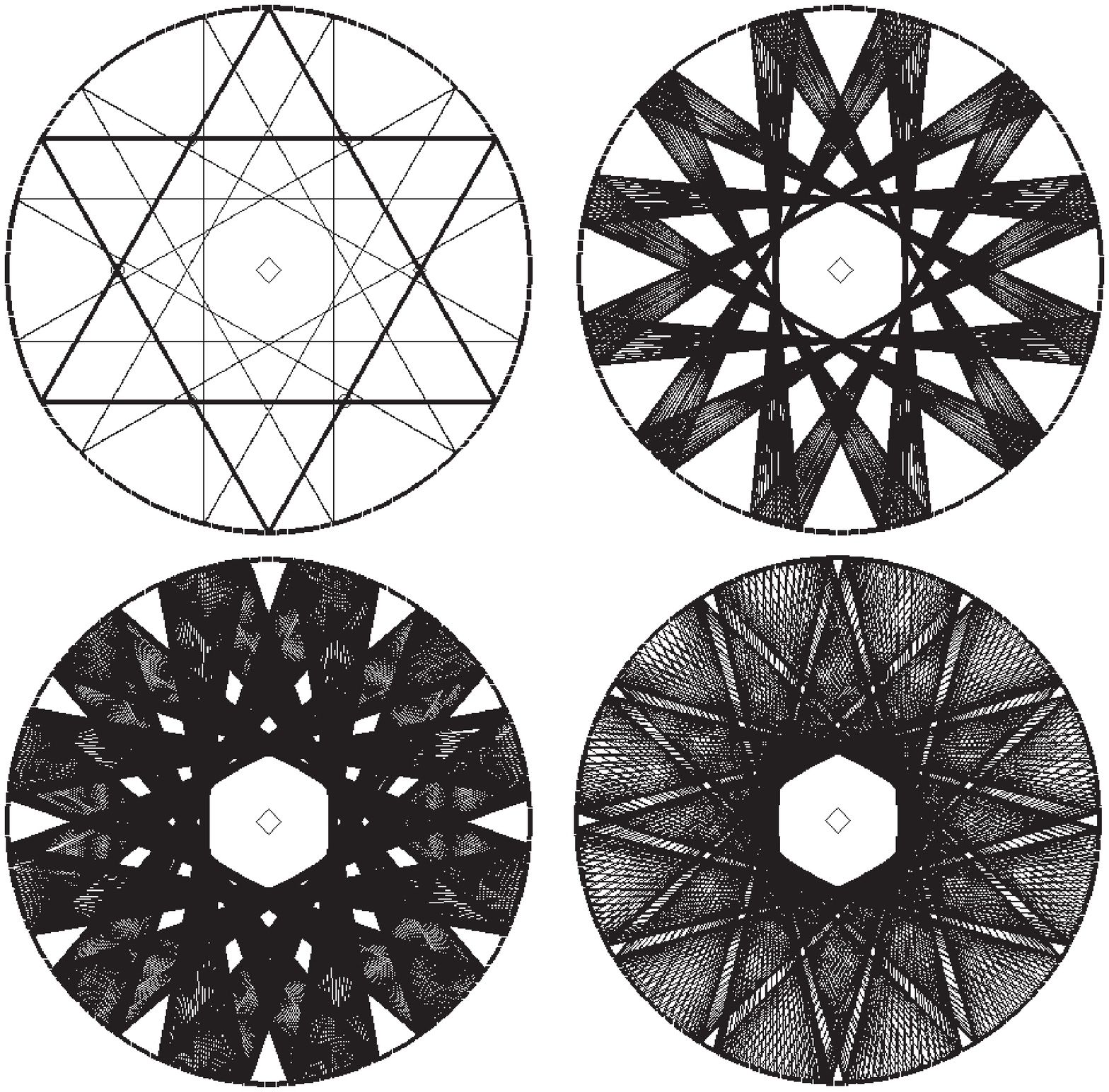}
\caption{Stable  ${12 \brace 5}$ orbit and several fat ${12 \brace 5}$'s}\label{lab11}
\end{figure}



\section{Forbidden inner region}\label{caustics}

In many of the figures (see for example Figs.~\ref{lab7},~\ref{lab8},~\ref{lab11}) we  
can  discern a  closed convex proper subset $\mathcal F$ of $\Omega$ 
which  lacks   segments of the corresponding billiard trajectory.  
 We shall call it a  \emph{``forbidden inner region''},  a term which we have
borrowed from Korsch and Zimmer \cite{Korsch2002}. It is separated from the allowed region
by  a non differentiable piecewise smooth closed curve.
By all of this we are, of course,  reminded of the notion of a caustic. There are several
different meanings of the word caustic  in billiards. The one that seems the most
appropriate  for us is (see Gruber~\cite{Gruber1990}, Klee and Wagon~\cite{Klee}, Turner~\cite{Turner1982}):

A convex caustic for a convex table $\Omega$  is a closed convex proper subset $\mathcal C$ of $\Omega$  
such that whenever the initial segment  of a billiard path in $\Omega$  lies in 
a supporting line of $\mathcal C$, then every later segment also lies in a supporting line of $\mathcal C$.
All billiard tables obtained by the string construction  have a special caustic. 
The hexagon $K$ is such  a caustic and we already established this in Theorem\ref{ppal}.

Are the  forbidden inner regions we get in the other cases  caustics as well?
For us to have a convex caustic every line segment of a trajectory   needs to touch  $C$.
However careful inspection  of the  billiard trajectories reveals that a few 
 line-segments  do not  touch $\mathcal C$ (see Fig.~\ref{lab12}).
 
  
 We  need to introduce an alternative description for the  forbidden inner region.
 We  define  the forbidden inner region $\mathcal F$ as   the intersection of all the ``left'' half-planes to the  
 lines containing the segments of  a billiard trajectory (see Hasselblatt and Katok~\cite{Hasselblatt}). 

We have included Fig.~\ref{lab12}  which allows us  to follow the construction process of a
forbidden inner region with an increasing number of half-planes.

\begin{figure}[htp]
\centering
\includegraphics[height=5.0in,width=5.0in]{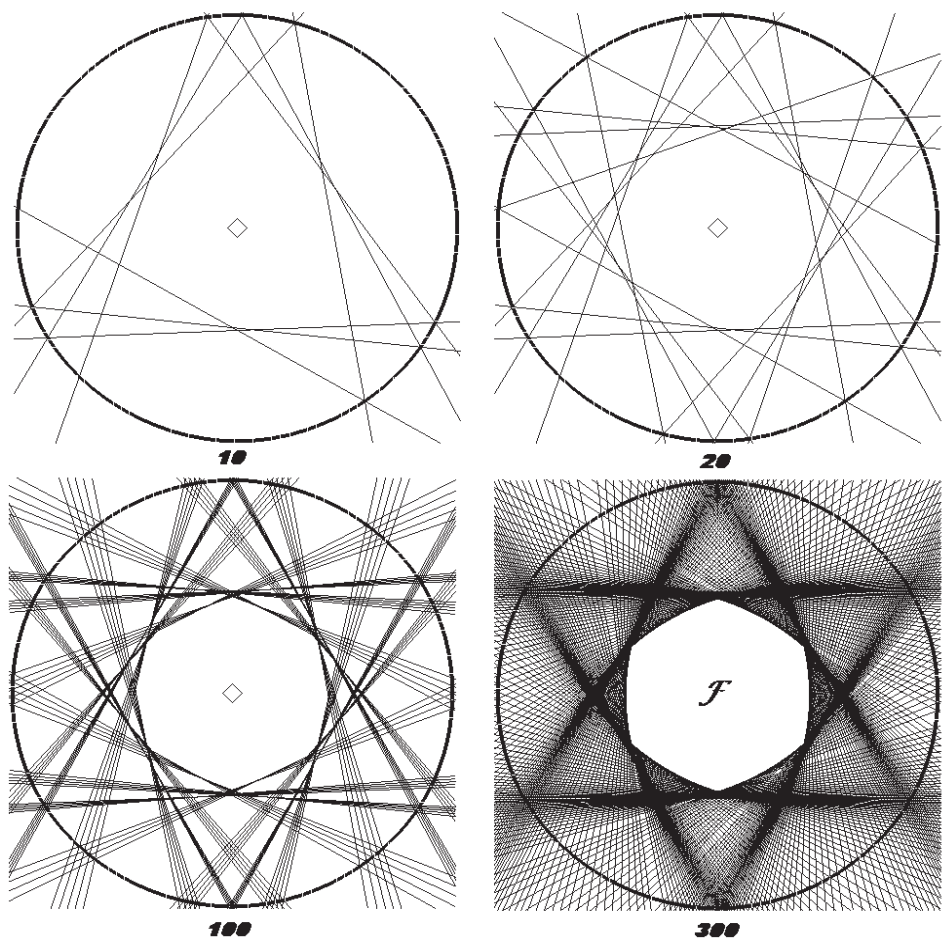}
\caption{Construction of forbidden inner region}\label{lab12}
\end{figure}




A problem remaining is that of trying to find an analytical expression for
 the piecewise smooth closed curve bounding  the forbidden inner region.
 The orbits shown in Figs. ~\ref{lab7},  ~\ref{lab8}, ~\ref{lab11} and ~\ref{lab13} strongly suggest
 that  a hexagon with curved  sides might be appropriate for the curve
 we are looking for. 
 
 \begin{figure}[htp]
\centering
\includegraphics[height=1.6in,width=6.0in]{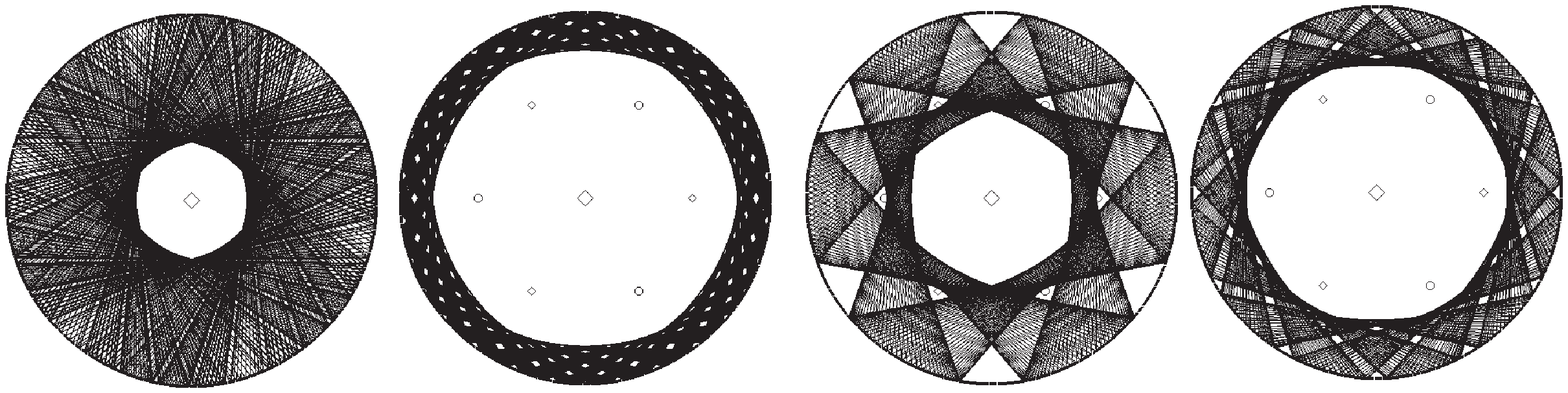}
\hspace{1.5in}\parbox{6in}{\caption{Forbidden inner regions for
groups of quasi-periodic orbits: $\qquad \qquad \qquad \qquad 6
\quad {5 \brace 2}'s,\quad 6 \quad {5 \brace 1}'s,\quad 2 \quad {3
\brace 1}'s, \quad \mbox{and} \quad 3 \quad {4 \brace
1}'s$\label{lab13}}}
\end{figure}

 Using   elliptical arcs for the  sides of the curved hexagon
 we get the results shown  in Fig.~\ref{lab14} and also in   Fig.~\ref{lab15}.

 \begin{figure}[htp]
\centering
\includegraphics[height=5.0in,width=5.0in]{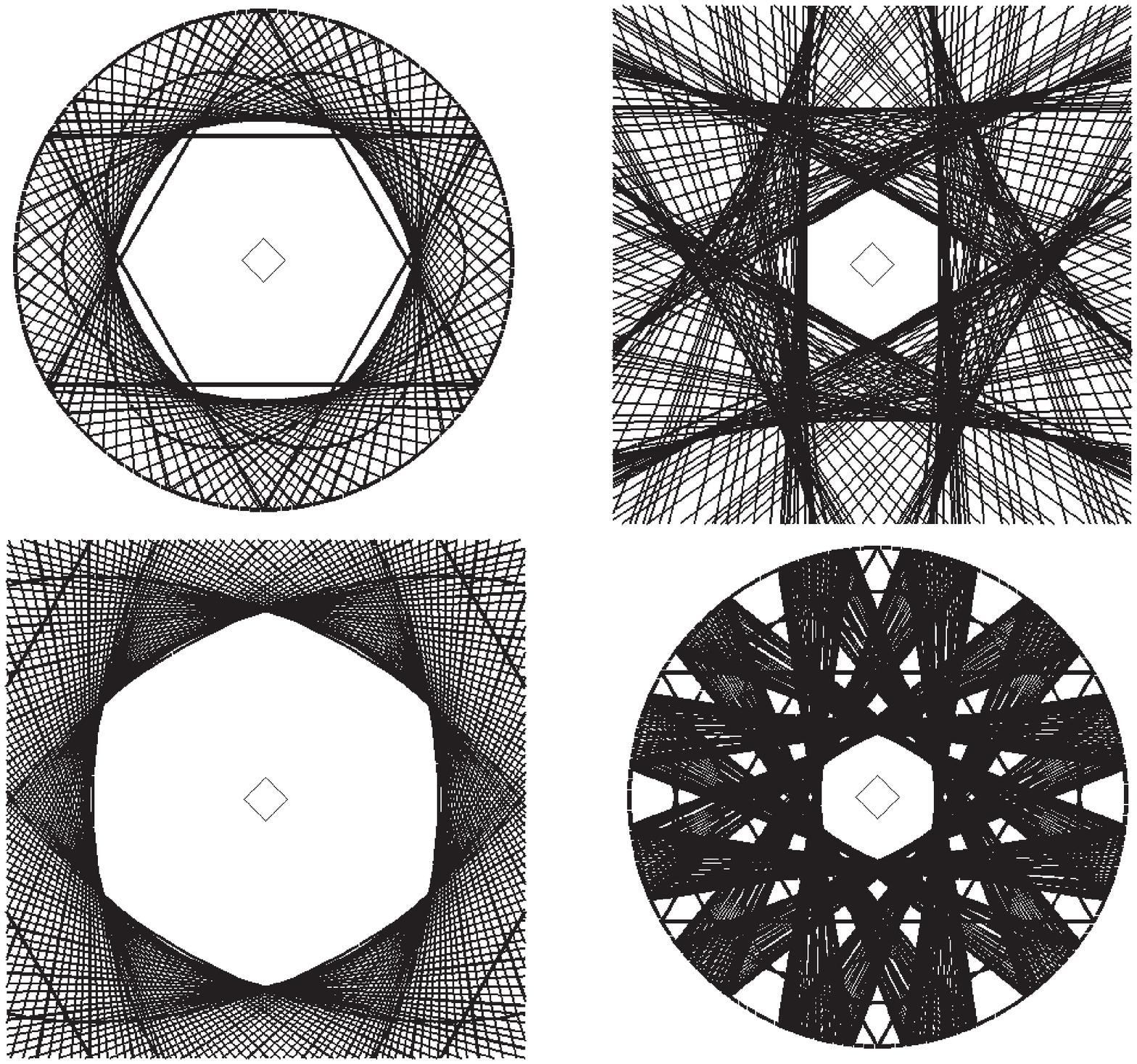}
\caption{Arcs of ellipses  separating  forbidden inner
region}\label{lab14}
\end{figure}

\begin{figure}[htp]
\centering
\includegraphics[height=1.8in,width=5.0in]{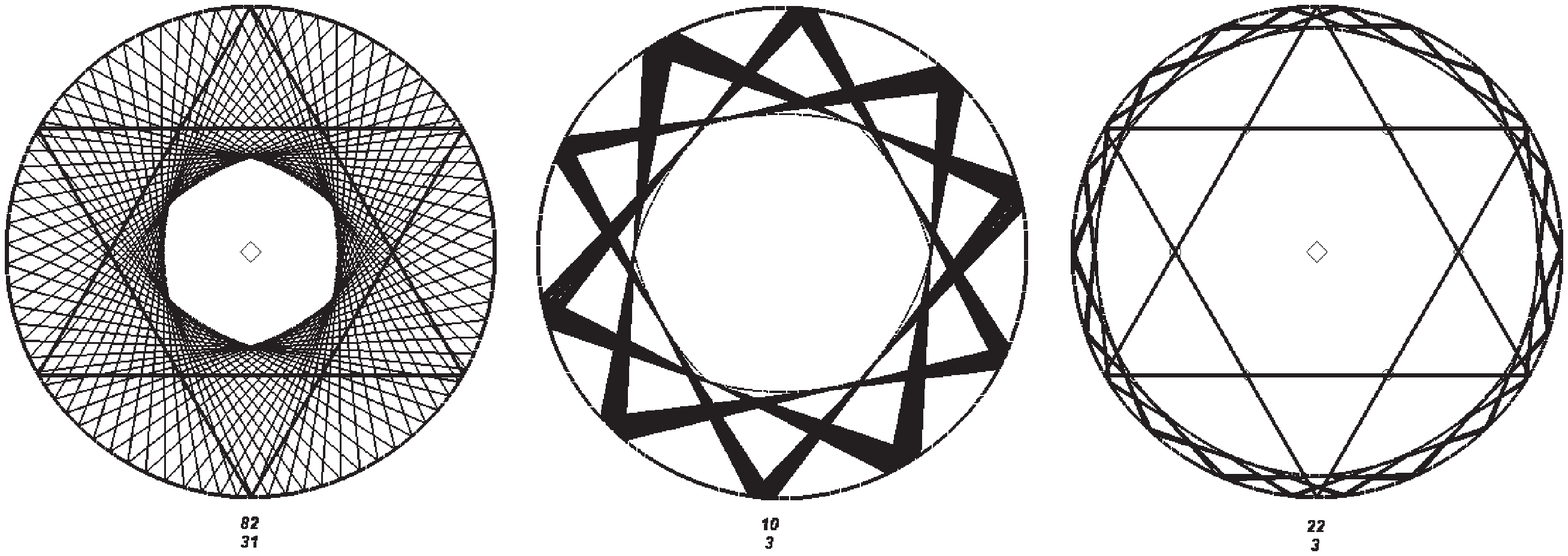}
\caption{Piecewise smooth curve  separating  forbidden inner
region}\label{lab15}
\end{figure}

 




 \section{Poincar\'{e} section}\label{poincare}

So far we have presented all orbits as merely sequences of directed
line segments in the configuration space. This approach, even if we
observe the evolution of a large number of these trajectories, gives
us only a limited understanding of the complicated dynamics
associated with the motion. A more convenient way to visualize the
dynamics is through the use of  Poincar\'e's surface of section
(SOS), which enables us to display the character of each particular
trajectory.
 The surface of section, see for instance Sussman and
Wisdom\cite{Sussman}, can be described by
 means of  two coordinates in the

\begin{enumerate}
\item {} horizontal direction: arc length $s$ from an arbitrarily chosen fixed  point $O$ on
the boundary of the billiard, and
\item{} vertical direction: the angle $\theta$ between positive
tangent direction and outgoing or reflected billiard ray at a point
of impact.

\end{enumerate}
This then defines the following   Poincar\'e section
$$\mathcal{P} := \{ (s,\theta) | s \epsilon [0,|\partial \Omega|], \theta \epsilon
 [0, \pi]\}$$

 There are, qualitatively speaking, three different
classes of trajectories possible on a SOS, clearly differentiated by
the dimension $d$ of the subspace of the section that they explore
(Berry\cite{Berry1981},Hayli\cite{Hayli1996},Korsch and Jodl\cite{Korsch},Sussman and
Wisdom\cite{Sussman}). They  might

\begin{enumerate}
\item{}  generate a finite set of $k$ discrete points
corresponding to a periodic orbit having period $k$ $(d = 0)$

 \item{} eventually fill out an invariant curve, which may either run
 from edge $s=0$ to edge $s=2\pi$ as an undulating line corresponding
 to a  whispering-gallery or a focal orbit,
 or consist of a group  of ovals or islands corresponding to a
 quasi-periodic orbit $(d = 1)$

 \item{} scatter over a  region eventually filling a whole area
 corresponding to an irregular or chaotic trajectory $(d = 2)$.

\end{enumerate}

Because our billiard has the same symmetry as a hexagon, we do not
need to  plot the  entire surface of section. We shall usually only
display the upper half of it and note the fact that the SOS is
periodic with period six.

To give the reader an idea of the  structure of the SOS we first
present an image in which we have included points and curves
corresponding to some 70 orbits, each followed through 480 bounces
(see Fig.~\ref{lab16}). In Fig.~\ref{lab17} we  do not show all
of them at once as was done before, but rather in small groups,  for
different ranges of the angle $\theta$.

\begin{figure}[htp]
\centering
\includegraphics[height=6.0in,width=6.0in]{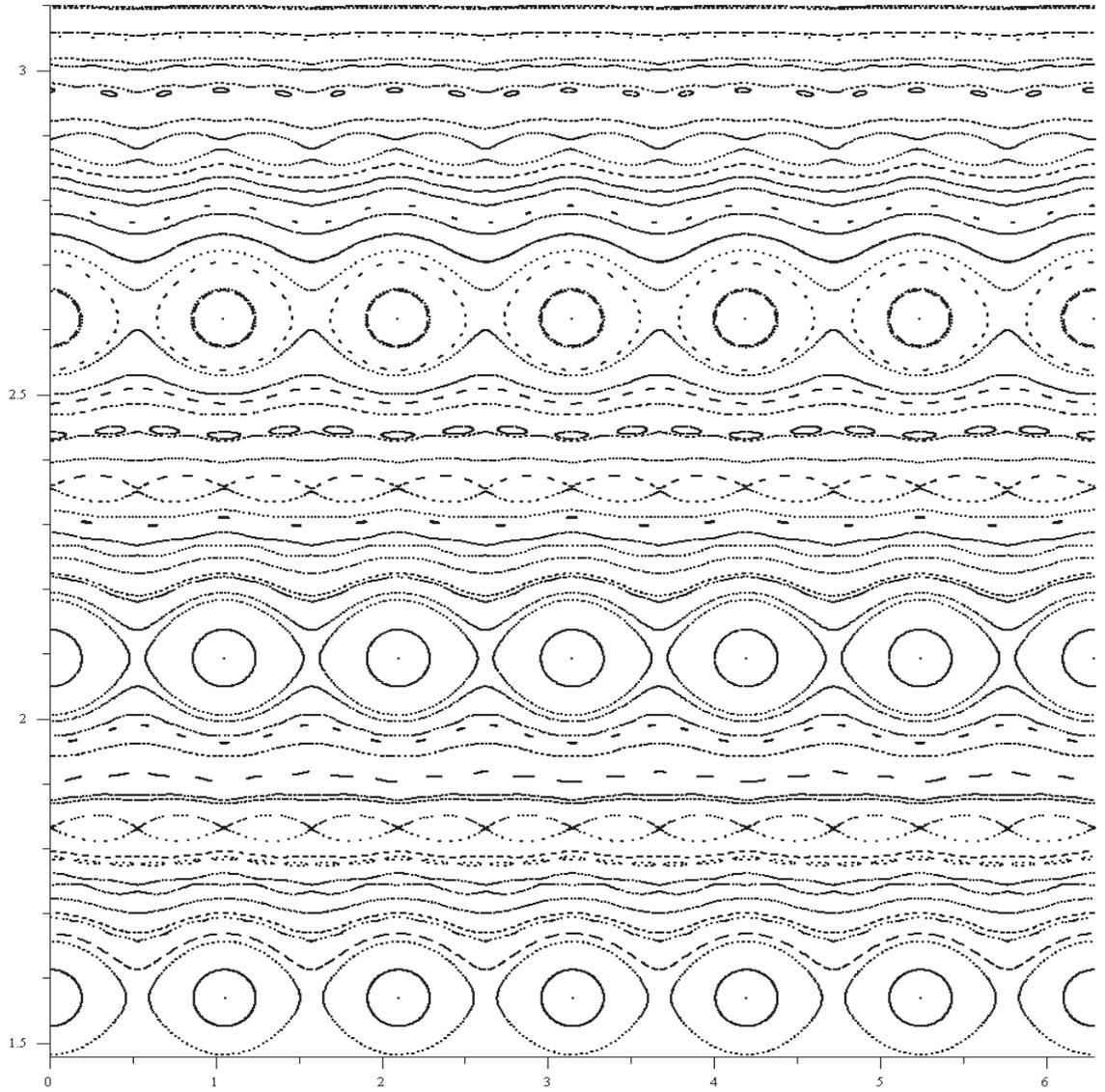}
\caption{Surface of section:
 70+ orbits followed through 480 bounces}\label{lab16}
\end{figure}


\begin{figure}[htp]
\centering
\includegraphics[height=6.0in,width=6.0in]{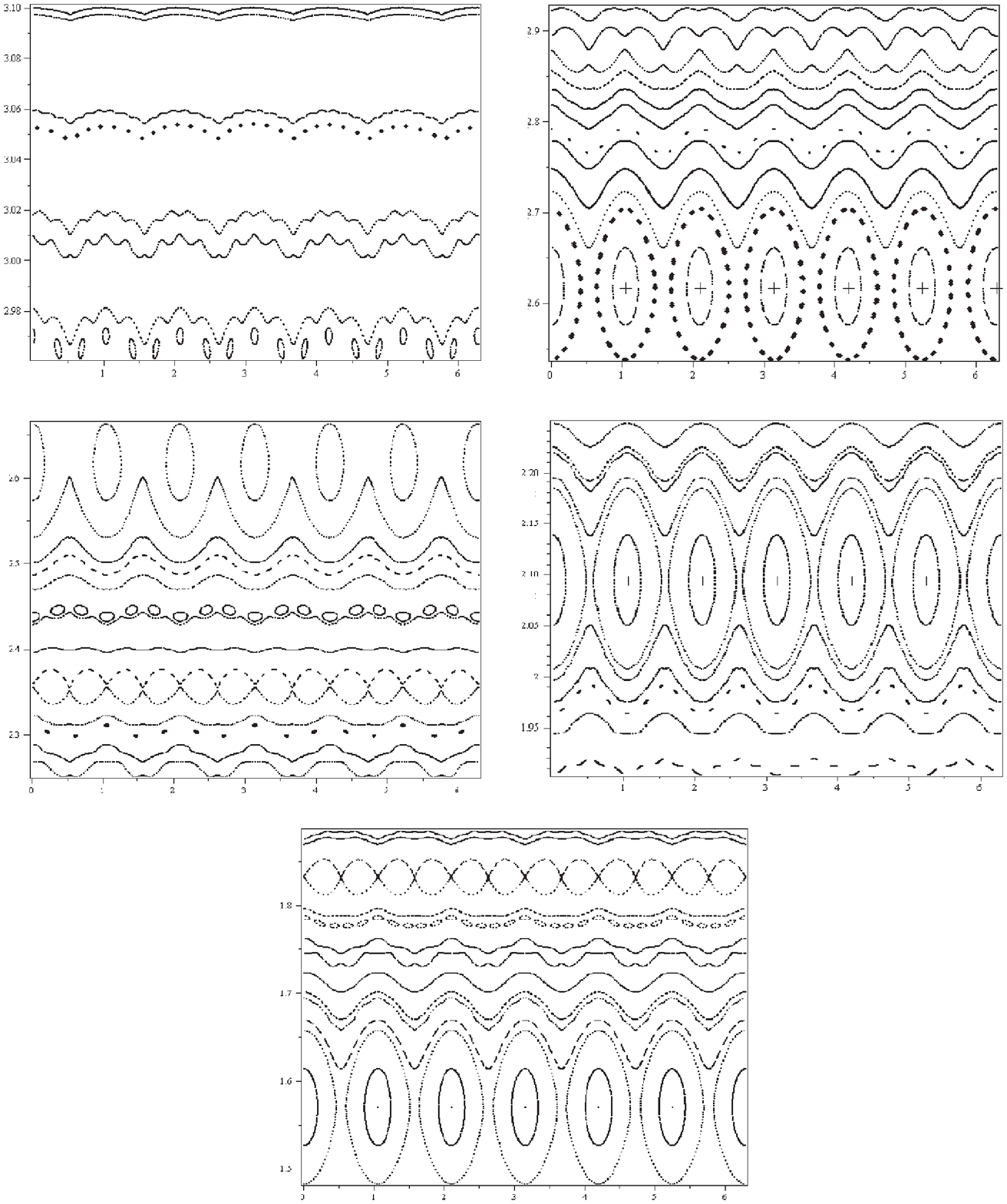}
\caption{Surface of section: same as Fig.~\ref{lab16}, but  more detailed
}\label{lab17}
\end{figure}


Several close-ups of the six-bounce orbits which can clearly be
recognized in  the plot on the right in the top row of Fig.~\ref{lab17}
 are shown in Fig.~\ref{lab18}.
 
\begin{figure}[htp]
\centering
\includegraphics[height=6.0in,width=6.0in]{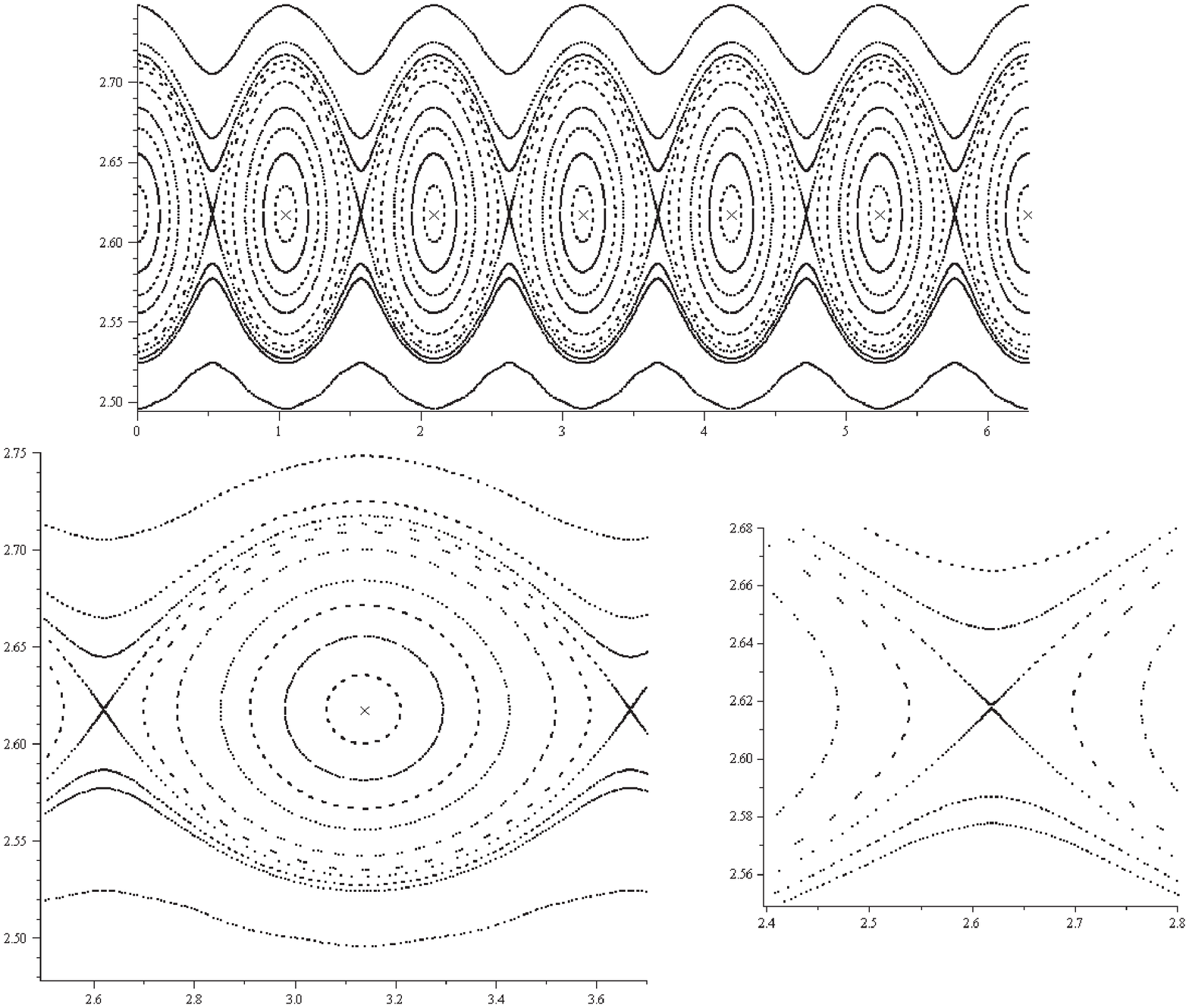}
\caption{ Surface of section: close-up of ${6 \brace 1}$
orbit}\label{lab18}
\end{figure}


Another close-up (Fig.~\ref{lab19}), this time of the bottom  plot
from Fig.~\ref{lab17} shows both undulating lines and ovals. The
big ovals surround the  invariant points corresponding
 to the stable periodic orbits associated with
${12 \brace 5},{54 \brace 23},{7 \brace 3}, \mbox{ and }{30 \brace
13}.$

 
 \begin{figure}[htp]
\centering
\includegraphics[height=5.0in,width=5.6in]{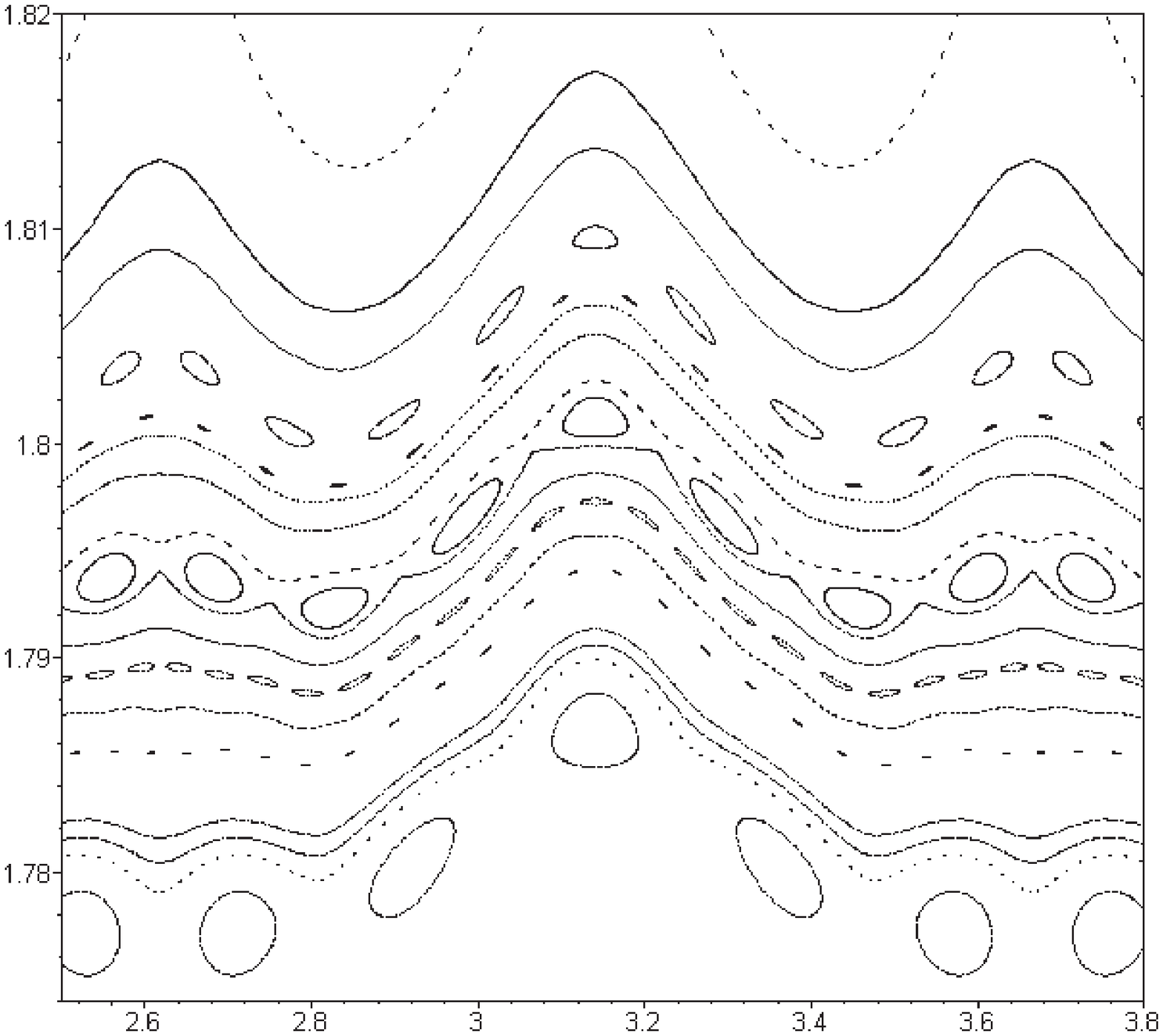}
\caption{ Surface of section  including fat ${12 \brace 5},{54 \brace
23},{7 \brace 3},  \mbox{and} {30 \brace 13}$} \label{lab19}
\end{figure}


Finally, we decided also to include the curve generated in the
surface of section by  a large number of focal orbits (see Fig.~\ref{lab20}). 
It is interesting to note that in this case it is even
possible to get an explicit expression for each arc that appears.
The coordinates of the points on a reference arc for the focal curve
are given by
$$ \left( E\left(\frac{\sqrt{(1+3\cos(2y))(-1+\cos(2y))}}{1-\cos(2y)},\frac{\sqrt{3}}{3}\right), y \right)$$
where $$\frac{\arccos(-1/3)}{2} \leq y \leq
\frac{\arccos(-1/2)}{2},$$

and  $ E(x,k)$ is simply an elliptic integral of the second kind
(see Davis \cite{Davis}).

\begin{figure}[htp]
\centering
\includegraphics[height=4.0in,width=4.0in]{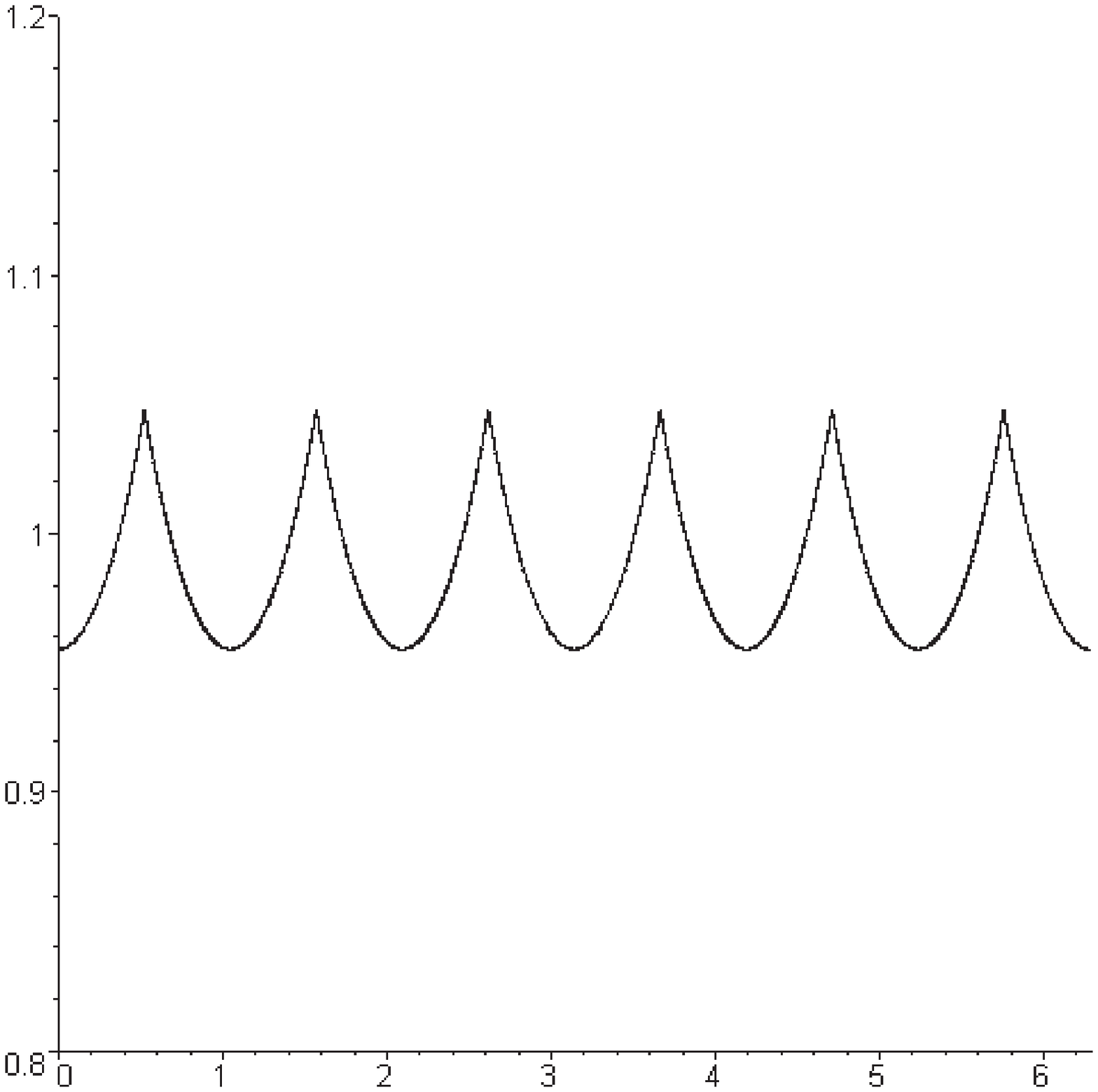}
\caption{ Surface of section for focal orbit}\label{lab20}
\end{figure}


\section{Additional periodic orbits and some comments on the structure of orbits in configuration space}

 All the periodic orbits we have  presented so far  are those
corresponding to  star polygons of the kind ${ n \brace k}$, 
with $1 \leq k < n/2$ and $(n,k) = 1,$ which  always come
in pairs:  one stable and the other one unstable.
 Moreover these orbits are isolated, as is customary
 for a generic billiard (Berry\cite{Berry1981}).

So, now, a natural question arises as to whether  there are any
other periodic orbits  besides those just mentioned.
 After all,   many authors, among them
Berry\cite{Berry1981},
Korsch and Zimmer\cite{Korsch2002},
 Okai et al.\cite{Okai1990}, and  Sieber \cite{Sieber1997}
 include periodic orbits in their work  which do not correspond to star polygons.
Of course we  found many additional periodic orbits for
our hexagonal string billiard of a different type, some of which are
shown in Figs. ~\ref{lab21} and ~\ref{lab22}.

In the captions of these figures  we have used the notation ${n \brace k}$ 
in spite of the fact that $n$ and $k$ are no
longer relatively prime, because we simply want to express the basic
fact that the orbit has period $n$ and that it winds around the
table $k$ times.

\begin{figure}[htp]
\centering
\includegraphics[height=5.0in,width=5.0in]{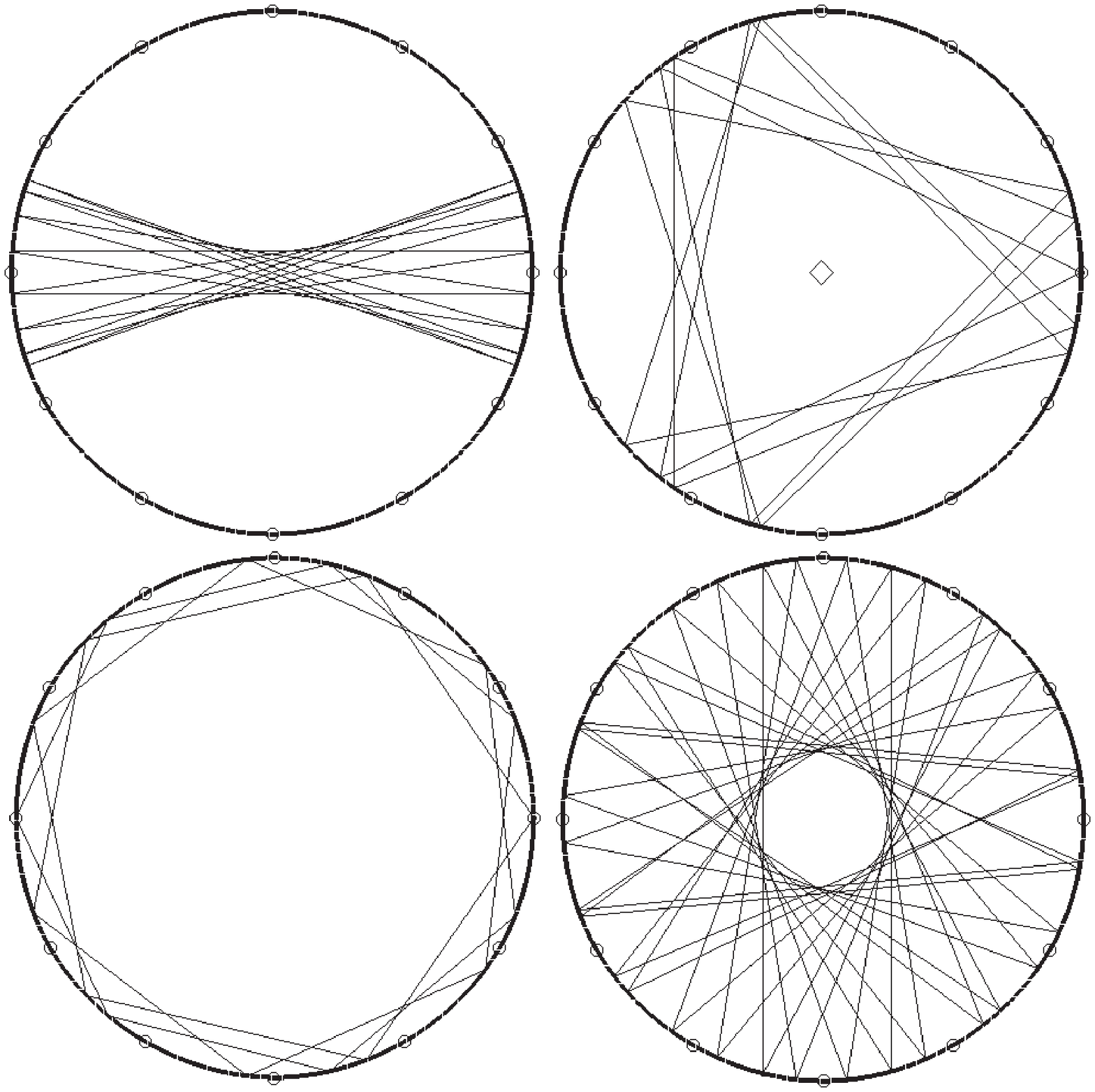}
\caption{ A librational, a $\left\{ 15 \atop 5 \right\}$, a $\left\{
18 \atop 3 \right\}$, and  a $\left\{ 36 \atop 15 \right\} $
orbit.}\label{lab21}
\end{figure}

\begin{figure}[htp]
\centering
\includegraphics[height=5.0in,width=5.0in]{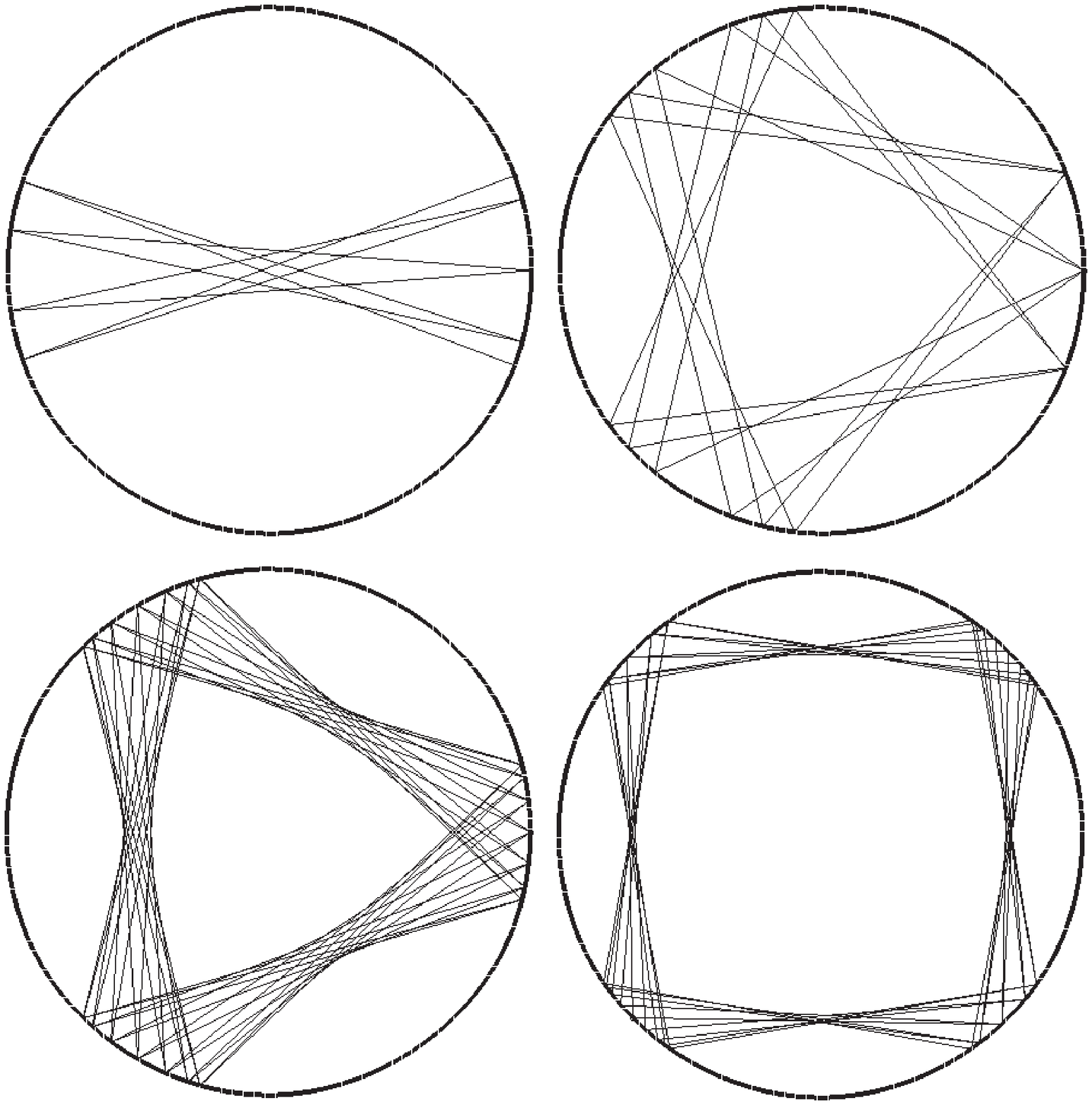}
\caption{ A self-retracing librational, a $\left\{ 18 \atop 6
\right\}$, a $\left\{ 42 \atop 14 \right\}$, and a $\left\{ 36 \atop
9 \right\}$ orbit.}\label{lab22}
\end{figure}


Interestingly enough, all of these orbits  are again isolated, but,
in contrast with the  ones corresponding to star polygons, these
additional periodic orbits turn out to be neutrally stable!


Now let us include a brief comment about the organization of
periodic and quasi-periodic trajectories in configuration space.

Periodic trajectories inside an elliptic billiard are very well organized.
Recall that for  fixed $n$ and $k$,  relatively prime positive  integers with $k < n/2$, 
there is  a  continuous family of periodic orbits of the form ${n \brace k}$. 
All the trajectories corresponding to such a family are arranged around 
a common ellipse or hyperbola (see Poncelet's theorem~\cite{Tabachnikov}).

All periodic  orbits for our billiard  are isolated. 
So for  fixed  $n$ and $k$,  where $n$ and $k$ are  relatively prime positive  integers with $k < n/2$,
we want to consider not only  the periodic orbits of the form ${n \brace k}$
but  all the quasi-periodic orbits or fat polygons ${n \brace k}$ as well.
If  there are any periodic orbits of the form ${m \brace j}$, $j < m/2$ with $m/j = n/k$
 we need to include  them too. 
We want to see how they  are organized in configuration space: if they are arranged around a
common forbidden inner region.

 To begin we will show what ``goes on between two consecutive'' unstable periodic
orbits:  ${4 \brace 1}$ and  ${3 \brace 1}$.

The claim now is that for either $n = 3$ or $4$ the pair of ${ n \brace 1 }$   unstable periodic
orbits ``encloses'' 
a  ${ n \brace 1 }$ stable periodic orbit, periodic orbits of the form ${m \brace j}$, $j < m/2$ with $m/j = n$
and fat polygons ${n \brace 1}$ as well.

In Fig.~\ref{lab23} we can appreciate a pair of unstable periodic
orbits ${4 \brace 1}$. Also included are a stable ${ 4 \brace 1}$ trajectory,
 a  neutrally stable periodic orbit ${36 \brace 9}$, as well as three quasi-periodic orbits.


\begin{figure}[htp]
\centering
\includegraphics[height=4.0in,width=6.0in]{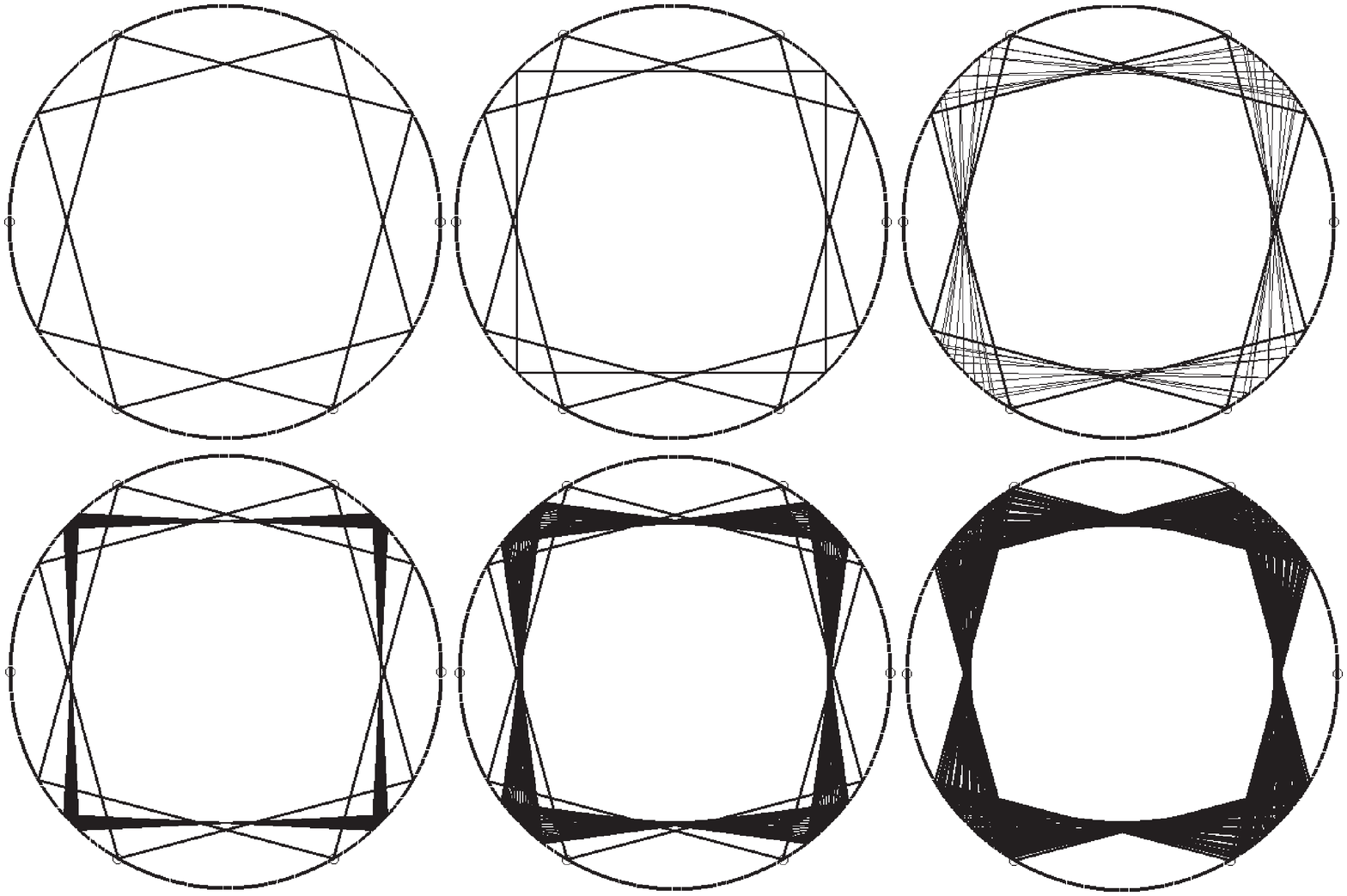}
\caption{ Trajectories related to the four-bounce periodic
orbit}\label{lab23}
\end{figure}




For  the second example we  start  with a pair of three-bounce
periodic unstable orbits. But now   we have also included a Reuleaux
type triangle (see Fig.~\ref{lab24}), which will
serve as a  common forbidden inner region.

For the three
neutrally stable periodic orbits in  Fig.~\ref{lab25} as
well as for the three quasi-periodic orbits we observe 
that they are confined to the region shown
on the right hand side of Fig.~\ref{lab24}.

\begin{figure}[htp]
\centering
\includegraphics[height=2.0in,width=4.0in]{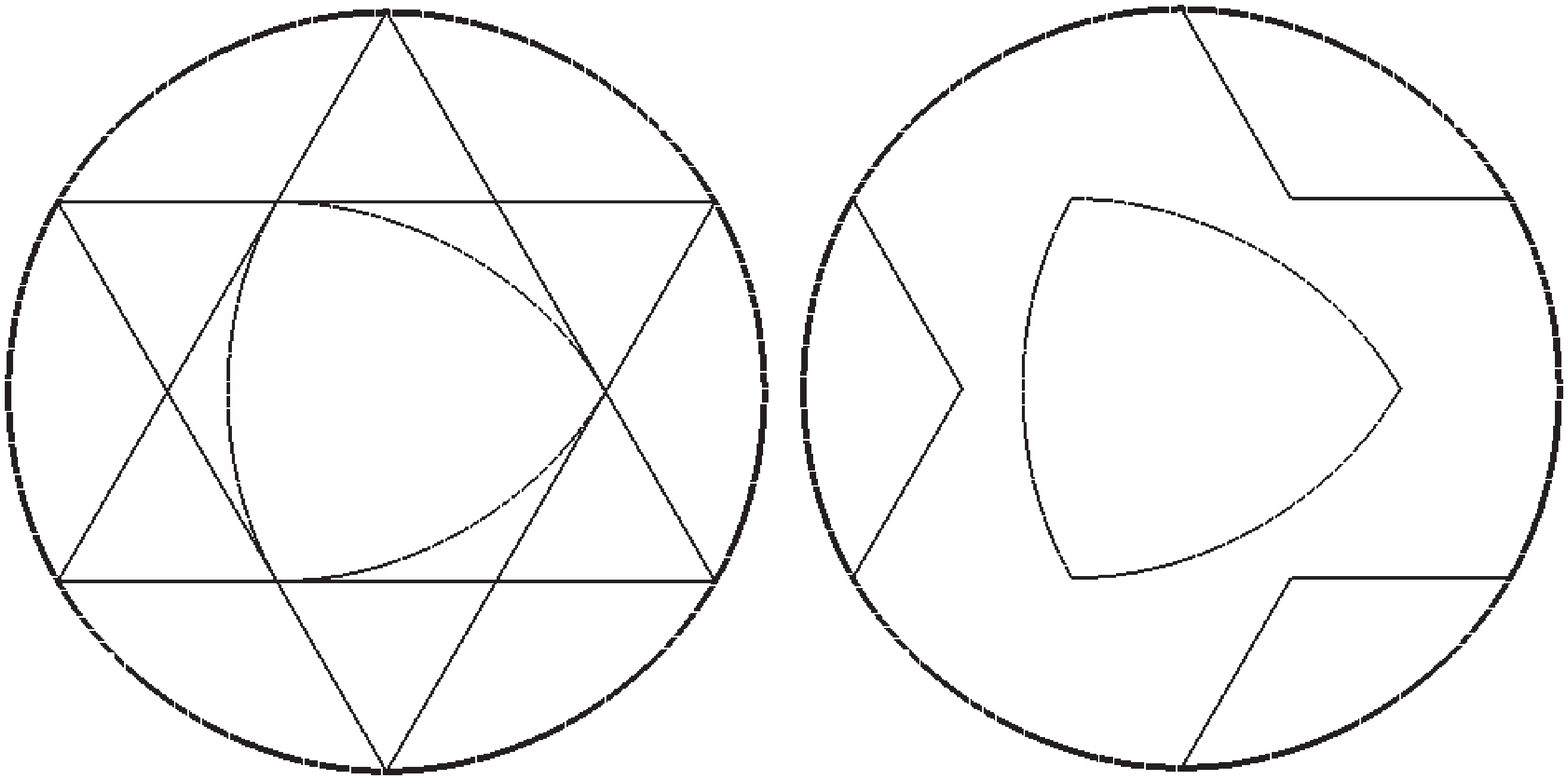}
\caption{ Unstable three bounce orbits and Reuleaux
triangle}\label{lab24}
\end{figure}


\begin{figure}[htp]
\centering
\includegraphics[height=4.0in,width=6.0in]{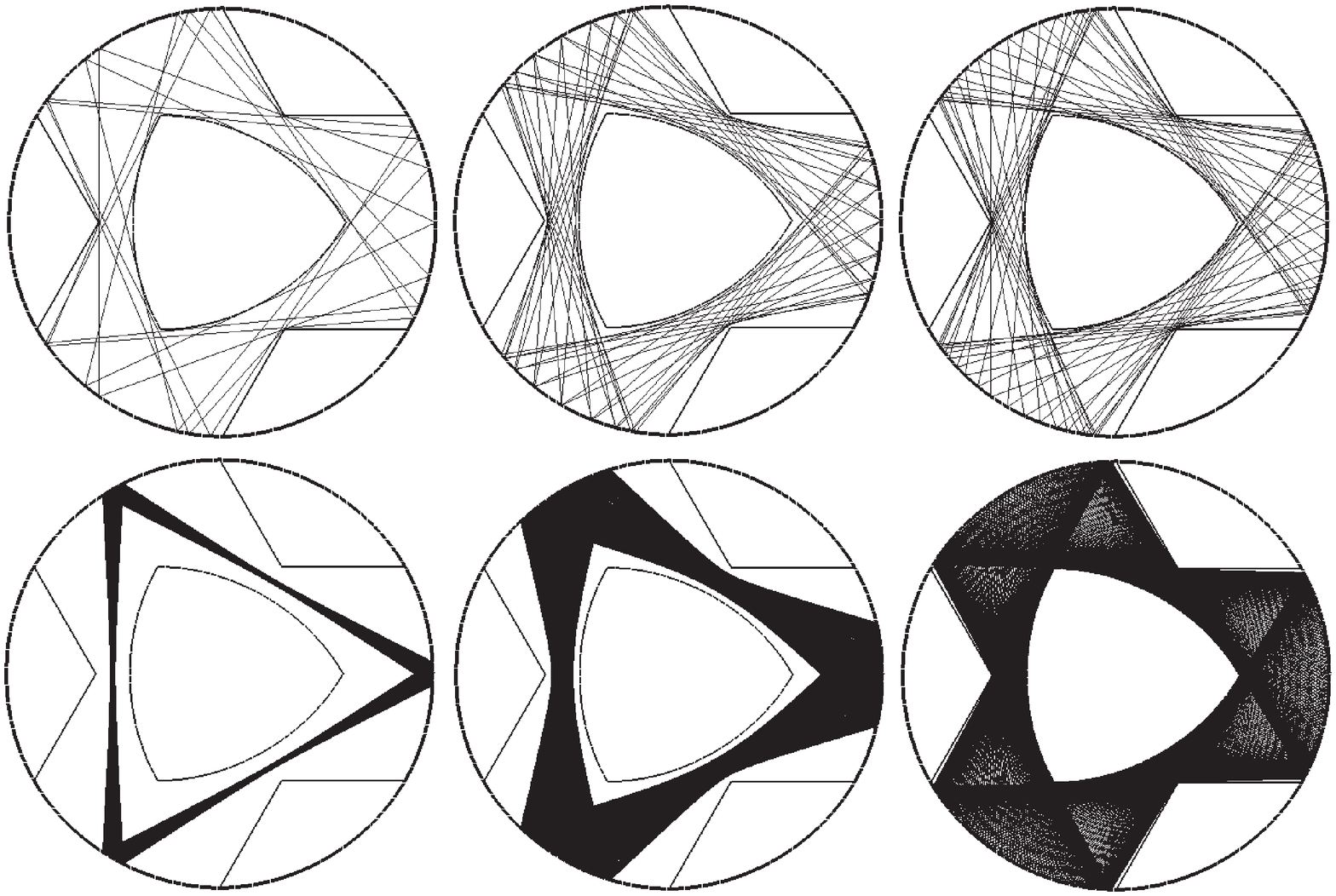}
\caption{ Trajectories related to the three-bounce periodic
orbit}\label{lab25}
\end{figure}


To illustrate the general situation  several pairs of consecutive
unstable periodic orbits  with a single quasi-periodic orbit
between them are shown in Fig.~\ref{lab26}.




\begin{figure}[htp]
\centering
\includegraphics[height=4.0in,width=6.0in]{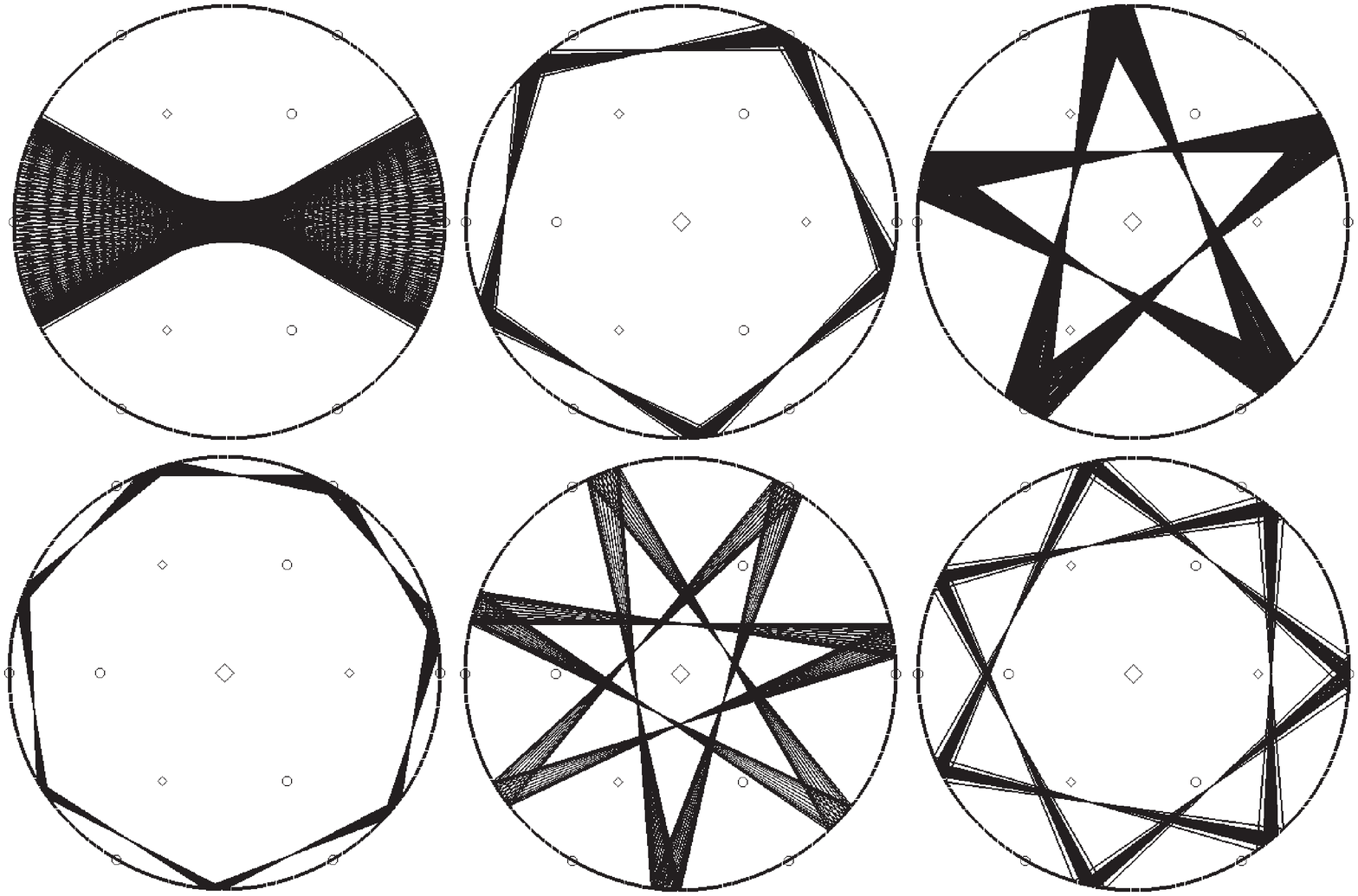}
\caption{ Pairs of unstable periodic orbits and single
quasi-periodic orbit between them }\label{lab26}
\end{figure}

To see how all the orbits related to the ${5 \brace 2}$ orbit are organized in 
configuration space see  Fig.~\ref{lab27}.

\begin{figure}[htp]
\centering
\includegraphics[height=3.8in,width=4.0in]{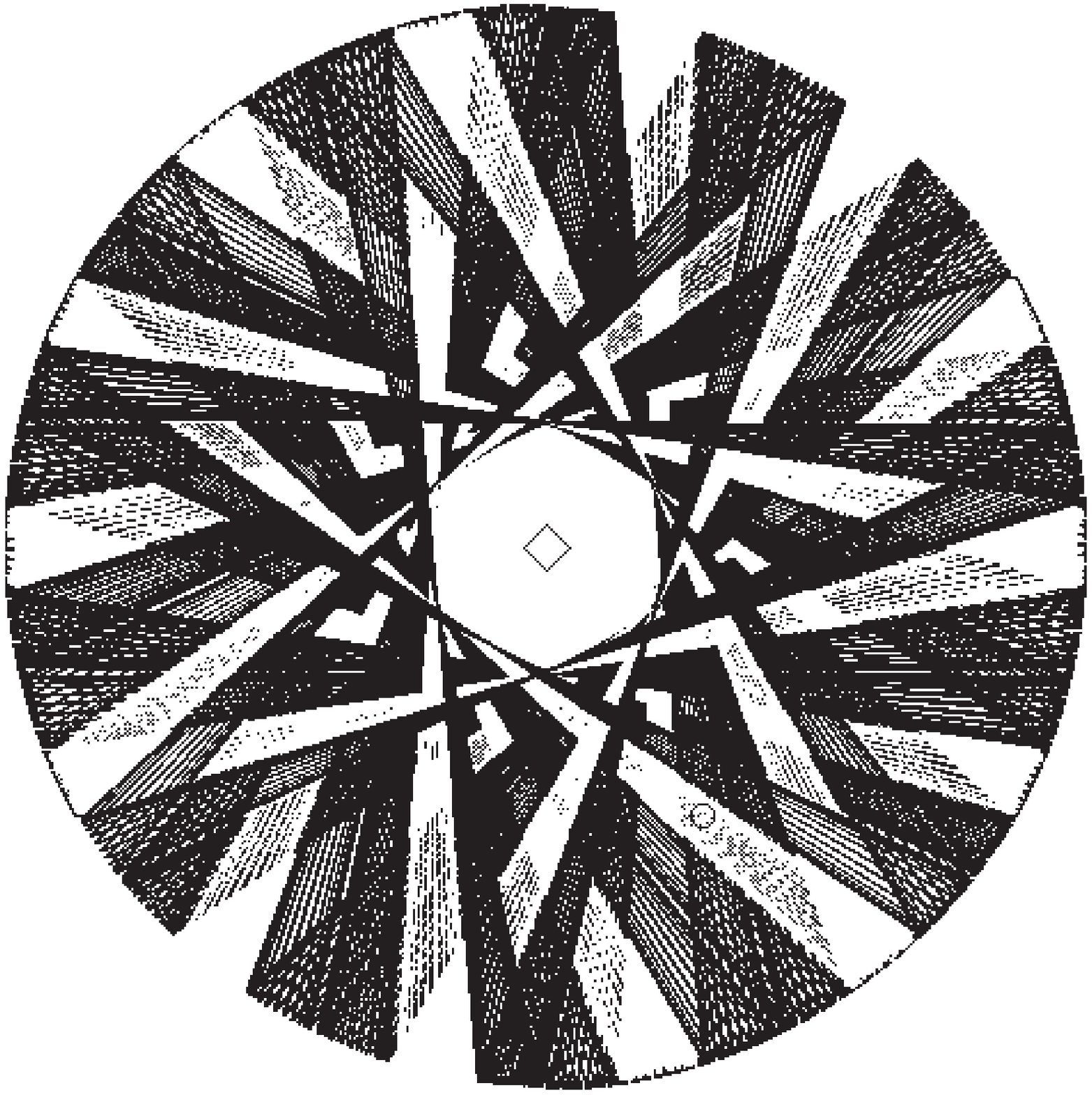}
\caption{${5 \brace 2}$ related orbits arranged around a
common forbidden inner region}\label{lab27}
\end{figure}

\section{Concluding remarks}

\begin{flushright}
 \begin{scriptsize}
EKELAND \cite{Ekeland}:  ``there is no
mistaking an integrable system  \\
with a nonintegrable one.''
 \end{scriptsize}
\end{flushright}


From the Poincar\'e surface of section one can get some pretty clear idea of the general
 behaviour a dynamical system. In it we can identify
 \begin{itemize}
 
 \item periodic trajectories  which appear as  a finite collection  
 of points, 
 \item quasiperiodic trajectories which  appear as  one-dimensional curves 

  \item and finally   irregular or 
chaotic trajectories which appear as a  scatter of points limited to a finite area.  
\end{itemize}
In the  Poincar\'e surface of section for a typical billiard we are able to find all of these different types of trajectories. 
A  Poincar\'e surface of section  consisting 
 exclusively of   closed curves and periodic points   suggests completely
 regular (integrable) behavior.
 
 The  extreme cases of completely regular  and of fully chaotic behavior 
are both a very uncommon occurrence.  Most
often one finds  that the  Poincar\'e SOS is filled with a
mixture of both regular and irregular orbits (see Korsch and Zimmer
\cite{ Korsch2002}).

So the question is:  exactly,  how exceptional is the hexagonal string billiard?

The following are the most relevant properties of this billiard:

\begin{enumerate}

\item{}
 The boundary of  our billiard
is twice continuously differentiable and  has strictly positive curvature.

\item{}
Our billiard is not a  simple deformation or perturbation of the circular billiard.
In addition to the planar region bounded by six elliptical arcs we also have a set 
of very special points: the foci, whose influence cannot be ignored. 

\item{}
Our billiard exhibits three types of trajectories: focal, outer, and
inner (see Theorem \ref{ppal}).

\item{}
It has a special caustic and we can associate a forbidden  region with any trajectory.

\item{}The  Poincar\'e surface of section is filled
solely with periodic points and invariant curves,  while chaotic area-filling orbits
are totally absent.

\item{}
In the Poincar\'{e}  surface of section there are no island chains.

\end{enumerate}

Finally let us briefly mention a  long-standing open problem, very often attributed to G. D. Birkhoff,
although  originally stated    by   Poritsky~\cite{Poritsky1950}(see also Gutkin~\cite{Gutkin2003}),
the so-called   Birkhoff-Poritsky  conjecture.
It basically states that among all billiards inside smooth closed convex curves, only billiards in ellipses are integrable. This means that aside from the energy, there must exist  a second  conserved quantity. 
Finding such a constant of motion  is usually a very difficult task. However 
in the case of the elliptic billiard it is well-known: it is simply the product  of the two angular momenta about the two foci (see Zhang et al.~\cite{Zhang1994}, Berry~\cite{Berry1981}, Korsch and Zimmer~\cite {Korsch2002}).

So, we want to  conclude with a question:

from all the evidence gathered for the hexagonal string  billiard, can we still be certain  that 
the elliptic billiard is the  only integrable convex billiard?

\section*{Acknowledgements}

For the  extensive programming task involved in this project we have
benefited greatly from  Carl Eberhart's efforts, who simulates
billiard-ball trajectories on both circular and elliptical tables in
\emph{Reflective paths in an ellipse}\cite{Eberhart}. We merely have
adapted some of his ideas for our own particular  purposes. We also
want to thank Martin Sieber for providing us, almost instantly, with
a hard to get and much needed reference.

\appendix
\section{}\label{A1}

\begin{proof}[Proof of Theorem~\ref{gardener}]
\mbox{}
Consider now the ellipses $\mathcal{E}_{1}$ with foci $F_{1}, F_{3}$  and $\mathcal{E}_{n}$ with  $F_{n},F_{2}$   and  where  the sum of the distances from the two foci is given by~$ l - 2(n-2) = 2+2d$.
  The length of the portion of the string that is in contact
with the $n$-gon is $2(n-2)$.

The ellipses $\mathcal{E}_{1}$   and $\mathcal{E}_{n}$  are of course very smooth
  closed curves and
  we need merely verify that the equation for $\mathcal{E}_{1}$ has the same function value and first and second derivatives as the equation for $\mathcal{E}_{n}$ at  $x=0$  thus ensuring $C^{2}$ continuity (see Fig.~\ref{lab2}).

Letting  $P(x,y)$ be an arbitrary point, then  for the ellipse $\mathcal{E}_{1}$:

$$\left|F_{1} P \right| +   \left|F_{3} P \right| =  2 - \frac{2}{\cos \alpha }$$ or

$$\sqrt{(x+1)^2+y^2}+\sqrt{(x-1+2 \cos \alpha)^2+(y+2 \sin \alpha)^2} = 2-\frac{2}{\cos \alpha} $$

Similarly, for   the ellipse $\mathcal{E}_{n}$:

$$\left|F_{n} P \right| +   \left|F_{2} P \right| =  2 - \frac{2}{\cos \alpha }$$  or

$$\sqrt{(x+1-2 \cos \alpha)^2+(y+2 \sin \alpha)^2}+\sqrt{(x-1)^2+y^2} = 2-\frac{2}{\cos \alpha}$$

 By  bringing all the terms from the right-hand side to the left-hand side,
  we get the  equations
  $$e_{1}(x,y) = 0 \quad \mbox{and} \quad  e_{n}(x,y) = 0$$

  Notice that since $$e_{n}(x,y) = e_{1}(-x,y)$$  the values of the functions and of the second derivatives for the  equations  of both   ellipses  have to coincide  when $x=0$.


 A simple, but tedious process of straightforward differentiation and substitution shows that
 at the common intersection point  $ G_{1} = (0,- \frac{\sin \alpha}{\cos \alpha})$ the first derivative  is $0$ and the second derivative  is $ (\cos \alpha-1)\cos \alpha\sin \alpha / (-1+2\cos \alpha)$ for both equations.

 \end{proof}
 
 
 \section{}\label{B1}
 \begin{proof}[Proof of Theorem~\ref{ppal}]
 \mbox{}
 \begin{enumerate}
 \item[i)]
 There are two extreme supporting lines at each vertex and 
 any  ray  that lies in the angle formed by the them is also a supporting line of  $K$.
 
 So to start, without loss of generality, let us  consider a billiard segment that lies in the angle formed by the
 two extreme supporting lines at vertex $F_{2}$. It will have necessarily one extreme point on arc $\stackrel{\frown}{24}$
 and the other on  arc $\stackrel{\frown}{62}$. If the extreme point on arc $\stackrel{\frown}{24}$
 is chosen as the  initial point $P_{0}$ 
 then we get a ray passing through the vertex $F_{2}$ which
 will  then intersect  the  arc $\stackrel{\frown}{62}$ at some point $P_{1}$.

Because of the  reflective property of ellipses a ray leaving one   focus, in this case
$F_{2}$, will reflect off the ellipse corresponding to the arc $\stackrel{\frown}{62}$ and 
pass through the second focus $F_{6}$. 

But then we get again a  ray (through $P_{1}$ and $F_{6}$) that lies in the angle formed by the   two extreme supporting lines for the vertex $F_{6}$ so it has to be on a supporting line itself.

 That just means that the billiard trajectory  remains on  a supporting line.
 
 \item[ii)]
 Let us now consider   the angle formed by the two supporting lines of $K$ through a point $P_{0}$ on arc $\stackrel{\frown}{24}$. One of them goes  through $F_{2}$ and the other through $F_{4}$. Because of this they
 are  also supporting lines of the closed segment $F_{2}F_{4}$, which is the segment between the two foci of the ellipse corresponding to arc $\stackrel{\frown}{24}$. 
 Now we can consider two
 different types of rays issuing from $P_{0}$: those strictly inside and those strictly outside the convex angle $F_{2}P_{0}F_{4}$. The first ones will intersect the  interiors of both $K$ and 
 of the segment $F_{2}F_{4}(see $\cite{Poorrezaei2003}), whereas the second ones  do not  intersect  either $K$ or the closed segment $F_{2}F_{4}$.
 
 In the first case we get a segment $P_{0}P_{1}$ where  $P_{1}$ is some point
belonging to the arc $\widehat{j,j+2}$. Let us  consider the reverse trajectory,
the one  starting at  point $P_{1}$ and ending at  point $P_{0}$. 
It, of course,  intersects the interior of $K$. But then it has to intersect the interior
 of the segment $F_{j}F_{j+2}$. So then the original trajectory
intersects this segment and it has to do so again after reflection with the boundary.
This implies that the reflected segment of the trajectory  then also intersects the  interior of $K$.
\item[iii)] We leave this case to the reader.
\end{enumerate}
 \end{proof}

 \section{}\label{C1}
 
 \begin{proof}[Proof of Theorem~\ref{angle}]
\mbox{}
Because of the symmetry present in our billiard,
we can restrict attention to a typical arc, say 
$\stackrel{\frown}{24}$.

Let $P(x,y)$ be  a point on this  arc (See Fig.~\ref{lab3}).
Then

$$\cos \gamma = \frac{-12+a^2+b^2}{2 a b},$$
where
$$ a= \sqrt{(x-1)^2+(y+\sqrt{3})^2},b=\sqrt{(x-1)^2+(y-\sqrt{3})^2}$$
so
 $$\cos \gamma =\frac{-12+2 (x-1)^2+(y+\sqrt{3})^2+(y-\sqrt{3})^2}{ 2 \sqrt{(x-1)^2+(y+\sqrt{3})^2} \sqrt{(x-1)^2+(y-\sqrt{3})^2}}$$
Substituting the equation for the ellipse 
$$ (x-1)^2=6(1-\frac{y^2}{9}) $$
we get
$$\cos \gamma = \frac{\frac{-4 y^2}{3}+(y+\sqrt{3})^2+(y-\sqrt{3})^2}{2   \sqrt{6-\frac{2 y^2}{3}+(y+\sqrt{3})^2} \sqrt{6-\frac{2 y^2}{3}+(y-\sqrt{3})^2}}$$
or after simplifying
$$ \cos \gamma = \frac{y^2+9}{27 -y^2}$$

Since $$-\sqrt{3} \leq  y  \leq \sqrt{3}$$ it is easy to see that
$$ \frac{1}{3} \leq \cos \gamma \leq \frac{1}{2} $$
so that 
 $$ 60^{\circ} \leq  \gamma \leq 70.53^{\circ}$$

\end{proof}
 \section{}\label{D1}

\begin{proof}[Proof of Theorem~\ref{triangle}]
\hfill \mbox{}
For the notation see Fig.~\ref{lab10}.
Consider the sequences  $\{s_{i}\},\{\varphi_{i}\},\{\alpha_{i}\}.$
It is clear  that they  are  bounded:
\begin{itemize}
\item  $ 2 \leq s_{i} \leq 4$
\item  $\pi/6 \leq \varphi_{i} \leq \pi/2$
\item $\pi/3  \leq \alpha_{i} \leq \arccos \left( \frac{1}{3} \right)$(see Thm.\ref{angle}).
\end{itemize}

We will show that   $\{\varphi_{i}\},\{s_{i}\}$ are both nondecreasing sequences.

Subtracting
$$\varphi_{i+1}+\varphi'_{i+1}=2\pi/3$$
from
$$\alpha_{i}+\varphi_{i}+\varphi'_{i+1}=\pi$$
we get
$$
0 \leq \alpha_{i}-\pi/3=\varphi_{i+1}-\varphi_{i}
$$
and therefore $\varphi_{i} \leq \varphi_{i+1}$.

From  
$$s_{i}^2 = 12+t_{i}^2-2 \cdot 2 \sqrt{3} t_{i}\cos \varphi_{i} $$
and $$s_{i} + t_{i} = 6$$
we get
$$s_{i} = \frac{4\sqrt{3}-6\cos \varphi_{i}}{\sqrt{3}-\cos \varphi_{i}}$$

Note that the function $f(\varphi) = \frac{4\sqrt{3}-6\cos \varphi}{\sqrt{3}-\cos \varphi}$
is nondecreasing.    

So if $\varphi_{i} \leq \varphi_{i+1}$ then $s_{i} = f(\varphi_{i}) \leq  f(\varphi_{i+1})=s_{i+1}$.

Now  $\{\varphi_{i}\}$ and  $\{s_{i}\}$ both converge  because 
they are monotone  bounded sequences.

But then as $i \rightarrow \infty$
$$\alpha_{i}-\pi/3=\varphi_{i+1}-\varphi_{i} \rightarrow 0$$
and therefore $$\alpha_{i}\rightarrow \pi/3$$

From
$$ 12 = s_{i}^2+t_{i}^2 - 2s_{i}t_{i}\cos \alpha_{i}$$
and 
$$ s_{i} + t_{i} = 6$$
we obtain
$$ \cos \alpha_{i} = -\frac{12+s_{i}^2-6 s_{i}}{s_{i}(-6+s_{i})}$$

So if $ i \rightarrow \infty$ then  $s_{i} \rightarrow  s^{\ast}$ and

$$ \frac{1}{2} = -\frac{12+{s^{\ast}}^2 -6 s^{\ast}}{s^{\ast}(-6+s^{\ast})}$$

which  yields $s^{\ast} = 2,4.$

Using the law of sines  we get
$$\frac{ \sin \alpha_{i} }{\sqrt{12}} = \frac{\sin \varphi_{i}}{s_{i}}  $$

 if   $ i \rightarrow \infty$ then

$$\frac{ {\sqrt{3}}/{2} }{\sqrt{12}} = \frac{\sin \varphi^{\ast}}{4}  $$

that is  $$ \sin \varphi^{\ast} = 1$$
and so finally 
$$ \varphi^{\ast} = \frac{\pi}{2}$$

\end{proof}

\end{document}